\documentclass[a4paper]{amsart}\usepackage[a4paper,margin=2.5cm,includefoot]{geometry}
\usepackage[utf8]{inputenc}
\usepackage[english]{babel}
\usepackage[T1]{fontenc}

\usepackage{amsmath,amssymb,mathtools,bm}
\usepackage{csquotes}
\usepackage{amstext}
\usepackage{amsthm}
\usepackage{bbm}
\usepackage{enumerate}
\usepackage{hyperref}
\usepackage[nameinlink,capitalise]{cleveref}
\usepackage{braket}
\usepackage{dsfont}
\usepackage{color}
\usepackage{tikz}
\usepackage{pgfplots}

\makeatletter
\def\@seccntformat#1{%
	\protect\textup{\protect\@secnumfont
		\ifnum\pdfstrcmp{subsection}{#1}=0 \bfseries\fi
		\ifnum\pdfstrcmp{subsubsection}{#1}=0 \itshape\fi
		\csname the#1\endcsname
		\protect\@secnumpunct
	}%
}
\renewcommand{\@upn}{}
\DeclareRobustCommand{\crefnosort}[1]{%
	\begingroup\@cref@sortfalse\cref{#1}\endgroup
}
\makeatother

\numberwithin{equation}{section}
\newtheorem{thm}{Theorem}[section]
\newtheorem{lem}[thm]{Lemma}
\newtheorem{prop}[thm]{Proposition}
\newtheorem{cor}[thm]{Corollary}

\theoremstyle{definition}

\newtheorem{hyp}{Hypothesis}
\renewcommand*{\thehyp}{\Alph{hyp}}

\theoremstyle{remark}
\newtheorem{rem}[thm]{Remark}

\crefname{hyp}{Hypothesis}{Hypotheses}
\Crefname{hyp}{Hypothesis}{Hypotheses}
\crefname{lem}{Lemma}{Lemmas}
\Crefname{lem}{Lemma}{Lemmas}
\crefname{thm}{Theorem}{Theorems}
\Crefname{thm}{Theorem}{Theorems}
\crefname{prop}{Proposition}{Propositions}
\Crefname{prop}{Proposition}{Propositions}
\crefname{enumi}{}{}
\Crefname{enumi}{}{}
\creflabelformat{enumi}{#2(#1)#3}
\crefname{equation}{}{}
\Crefname{equation}{}{}
\crefname{rem}{Remark}{Remarks}
\Crefname{rem}{Remark}{Remarks}

\makeatletter
\renewcommand{\@upn}{} 
\makeatother

\usepackage[inline]{enumitem}
\newlist{enumthm}{enumerate}{1} 
\setlist[enumthm]{label=\upshape(\roman*),ref=\thethm\,(\roman*)}  
\crefalias{enumthmi}{thm} 
\newlist{enumcor}{enumerate}{1}
\setlist[enumcor]{label=\upshape(\roman*),ref=\thecor\,(\roman*)}
\crefalias{enumcori}{cor}
\newlist{enumlem}{enumerate}{1}
\setlist[enumlem]{label=\upshape(\roman*),ref=\thelem\,(\roman*)}
\crefalias{enumlemi}{lem}
\newlist{enumprop}{enumerate}{1}
\setlist[enumprop]{label=\upshape(\roman*),ref=\theprop\,(\roman*)}
\crefalias{enumpropi}{prop}
\newlist{enumhyp}{enumerate}{1}
\setlist[enumhyp]{label=\upshape(\roman*),ref=\thehyp\,(\roman*)}
\crefalias{enumhypi}{hyp}
\newlist{enumproof}{enumerate*}{1}
\setlist[enumproof]{label=\upshape(\roman*)}
\newlist{enumdef}{enumerate}{1}
\setlist[enumdef]{label=\upshape(\roman*),ref=\thedefn\,(\roman*)}
\crefalias{enumdefi}{defn}

\makeatletter
\newcounter{subcreftmpcnt} %
\newcommand\romansubformat[1]{(\roman{#1})} 
\def\subcref{\@ifstar\@@subcref\@subcref}
\newcommand\@subcref[2][\romansubformat]{%
	\ifcsname r@#2@cref\endcsname
	\cref@getcounter {#2}{\mylabel}%
	\setcounter{subcreftmpcnt}{\mylabel}%
	\hyperref[#2]{\romansubformat{subcreftmpcnt}}%
	\else ?? \fi}   
\newcommand\@@subcref[2][\romansubformat]{%
	\ifcsname r@#2@cref\endcsname
	\cref@getcounter {#2}{\mylabel}%
	\setcounter{subcreftmpcnt}{\mylabel}%
	\romansubformat{subcreftmpcnt}%
	\else ?? \fi}   
\makeatother

\makeatletter
\DeclareRobustCommand{\crefnosort}[1]{%
	\begingroup\@cref@sortfalse\cref{#1}\endgroup
}
\makeatother

\def\endstepsymbol{$\lozenge$}
\def\endclaimsymbol{$\lozenge$}
\newcounter{proofstep}
\AtBeginEnvironment{proof}{\setcounter{proofstep}{0}}

\crefname{proofstep}{Step}{Steps}
\Crefname{proofstep}{Step}{Steps}
\newcounter{proofclaim}
\AtBeginEnvironment{proof}{\setcounter{proofclaim}{0}}

\crefname{proofclaim}{Claim}{Claims}
\Crefname{proofclaim}{Claim}{Claims}

\newcommand{\cB}{{\mathcal B}}
\newcommand{\cF}{{\mathcal F}}
\newcommand{\cH}{{\mathcal H}}
\newcommand{\cK}{{\mathcal K}}
\newcommand{\cM}{{\mathcal M}}

\newcommand{\fS}{{\mathfrak S}}

\newcommand{\fb}{{\mathfrak b}}

\newcommand{\fh}{{\mathfrak h}}

\newcommand{\BC}{{\mathbb C}}

\newcommand{\BN}{{\mathbb N}}
\newcommand{\BR}{{\mathbb R}}

\usepackage{dsfont} 

\newcommand{\dsone}{{\mathds 1}}

\usepackage{mathrsfs}

\newcommand{\sD}{{\mathscr D}}

\newcommand{\sfB}{{\mathsf B}}
\newcommand{\sfD}{{\mathsf D}}\newcommand{\sfE}{{\mathsf E}}

\newcommand{\sfN}{{\mathsf N}}

\newcommand{\sfS}{{\mathsf S}}

\newcommand{\sfd}{{\mathsf d}}\newcommand{\sfe}{{\mathsf e}}\newcommand{\sff}{{\mathsf f}}
\newcommand{\sfg}{{\mathsf g}}\newcommand{\sfi}{{\mathsf i}}

\newcommand{\sfn}{{\mathsf n}}
\newcommand{\sfr}{{\mathsf r}}
\newcommand{\sfs}{{\mathsf s}}

\newcommand{\rme}{{\mathrm e}}

\newcommand{\IN}{\BN}\newcommand{\IR}{\BR}\newcommand{\IC}{\BC}
\newcommand{\C}{\BC}

\newcommand{\hs}{\fh}\newcommand{\HS}{\cH}

\newcommand{\eps}{\varepsilon}\newcommand{\ph}{\varphi}

\newcommand{\e}{\rme}\newcommand{\ii}{\sfi}\newcommand{\Id}{\dsone} \renewcommand{\d}{\sfd}

\newcommand{\ran}{\operatorname{ran}}

\DeclareMathOperator*{\slim}{s-lim}

\DeclareFontFamily{U}{mathx}{\hyphenchar\font45}
\DeclareFontShape{U}{mathx}{m}{n}{
	<5> <6> <7> <8> <9> <10>
	<10.95> <12> <14.4> <17.28> <20.74> <24.88>
	mathx10
}{}
\DeclareSymbolFont{mathx}{U}{mathx}{m}{n}
\DeclareFontSubstitution{U}{mathx}{m}{n}
\DeclareMathAccent{\widecheck}{0}{mathx}{"71}
\DeclareMathAccent{\wideparen}{0}{mathx}{"75}

\DeclareFontFamily{OMX}{MnSymbolE}{}
\DeclareFontShape{OMX}{MnSymbolE}{m}{n}{
	<-6>  MnSymbolE5
	<6-7>  MnSymbolE6
	<7-8>  MnSymbolE7
	<8-9>  MnSymbolE8
	<9-10> MnSymbolE9
	<10-12> MnSymbolE10
	<12->   MnSymbolE12}{}
\DeclareSymbolFont{mnlargesymbols}{OMX}{MnSymbolE}{m}{n}
\SetSymbolFont{mnlargesymbols}{bold}{OMX}{MnSymbolE}{b}{n}
\DeclareMathDelimiter{\llangle}{\mathopen}{mnlargesymbols}{'164}{mnlargesymbols}{'164}
\DeclareMathDelimiter{\rrangle}{\mathclose}{mnlargesymbols}{'171}{mnlargesymbols}{'171}
\DeclareMathDelimiter{\lsem}{\mathopen}{mnlargesymbols}{'102}{mnlargesymbols}{'102}
\DeclareMathDelimiter{\rsem}{\mathclose}{mnlargesymbols}{'107}{mnlargesymbols}{'107}
\DeclareMathDelimiter{\langlebar}{\mathopen}{mnlargesymbols}{'152}{mnlargesymbols}{'152}
\DeclareMathDelimiter{\ranglebar}{\mathclose}{mnlargesymbols}{'157}{mnlargesymbols}{'157}
\DeclareMathDelimiter{\lWavy}{\mathopen}{mnlargesymbols}{'137}{mnlargesymbols}{'137}
\DeclareMathDelimiter{\rWavy}{\mathopen}{mnlargesymbols}{'137}{mnlargesymbols}{'137}

\newcommand{\chr}{\mathbf 1}
\newcommand{\abs}[1]{\lvert#1\lvert}
\newcommand{\norm}[1]{\lVert#1\lVert}\newcommand{\Norm}[1]{\left\lVert#1\right\lVert}

\newcommand{\FGamma}{\Gamma}
\newcommand{\FS}{\cF}\newcommand{\dG}{\sfd\FGamma}\newcommand{\ad}{a^\dagger}

\title[Ultraviolet Renormalization of Spin Boson Models I]{Ultraviolet Renormalization of Spin Boson Models I.\\ Normal and 2-Nilpotent Interactions}
\author{Benjamin Hinrichs}
\address{B. Hinrichs and J. Valent\'in Mart\'in, Universit\"at Paderborn, Institut f\"ur Mathematik, Institut f\"ur Photonische Quantensysteme, Warburger Str. 100, 33098 Paderborn, Germany}
\email{benjamin.hinrichs@math.upb.de}
\email{javiervm@math.uni-paderborn.de}
\author{Jonas Lampart}
\address{J. Lampart, CNRS \& Laboratoire Interdisciplinaire Carnot de Bourgogne (UMR 6303), Université de Bourgogne, 9 Av. A. Savary, 21078 Dijon Cedex, France.}
\email{jonas.lampart@u-bourgogne.fr}
\author{Javier Valent\'in Mart\'in}

\usepackage{xcolor}\usepackage{cancel}

\newcommand{\HSB}{H_{\sfS\sfB}}
\newcommand{\HSBL}{H_{\sfS\sfB,\Lambda}}
\newcommand{\Hreg}{H_{\sfr\sfe\sfg}}

\newcommand{\HIBCl}{\HIBC{\lambda}}

\renewcommand{\ad}{a^*}

\begin{document}

\begin{abstract} 
	\noindent
	We study the ultraviolet problem for models of a finite-dimensional quantum mechanical system linearly coupled to a bosonic quantum field, such as the (many-)spin boson model or its rotating-wave approximation.
	If the state change of the system upon emission or absorption of a boson is either given by a normal matrix or by a 2-nilpotent one, which is the case for the previously named examples, we prove an optimal renormalization result.
	We complement it, by proving the norm resolvent convergence of appropriately regularized models to the renormalized one.
	Our method consists of a dressing transformation argument in the normal case and an appropriate interior boundary condition for the 2-nilpotent case.
\end{abstract}

\maketitle

\section{Introduction}

The spin boson model describes a two-state quantum mechanical system, called spin, linearly coupled to a bosonic quantum field.
It is the most simple non-trivial interacting model involving a quantum field. Despite this simplicity, it has found various applications in quantum physics, ranging from open quantum systems and quantum optics to quantum information and solid state physics, whence it is an active field of research up to date. Since the number of references is too numerous, let us merely point to the reviews \cite{Leggettetal.1987,HuebnerSpohn.1995b} as well as the reference sections of the articles cited in the following.

As far as the mathematical literature on the spin boson model is concerned, significant attention has been devoted to studying its infrared properties. This investigation was started in \cite{Spohn.1989,Amann.1991,AraiHirokawa.1995}, amongst others continued in \cite{HirokawaHiroshimaLorinczi.2014,BallesterosDeckertHaenle.2019b} and has especially produced the remarkable observation that infrared divergences cancel at small coupling \cite{HaslerHerbst.2010,BachBallesterosKoenenbergMenrath.2017,HaslerHinrichsSiebert.2021a} but not at large coupling \cite{BetzHinrichsKraftPolzer.2025}.
Again, we emphasize that the selected literature referenced here is far from a complete list.

On the other side, the ultraviolet problem for the spin boson model has only recently come into focus \cite{DamMoller.2018b,Lonigro.2022,Lonigro.2023,LillLonigro.2024},
motivated by the many other applications than the quantum field theoretic setting.
The attention is hereby not restricted to the two-state spin, but to the wider setting of generalized spin boson models, first treated in \cite{AraiHirokawa.1997}.
In the present article, we provide an optimal renormalization result for a wide class of these generalized spin boson models and beyond.

In the context of non- and semi-relativistic quantum field theory, ultraviolet renormalization has been addressed using a variety of methods.
Nelson \cite{Nelson.1964} applied a dressing transformation to study the model of a non-relativistic particle linearly coupled to a bosonic field.
His method received quite a few refinements in later years \cite{Cannon.1971,Ammari.2000,GriesemerWuensch.2018,DamHinrichs.2021}, was also applied to the Fr\"ohlich polaron model \cite{GriesemerWuensch.2016}, and is similar to one part of our renormalization technique.
In fact, since the spin in our model does not move, the dressing transformation applied here is also close to older results on the van Hove model \cite{vanHove.1952}, see for example \cite{Derezinski.2003} for an extensive discussion.
A different approach applying interior boundary conditions (IBC), inspired by the treatment of point interactions \cite{Yafaev.1992a,LampartSchmidtTeufelTumulka.2018}, was also recently applied to Nelson's model \cite{LampartSchmidt.2019,Schmidt.2021} as well as a two-dimensional semi-relativistic version of Nelson's model \cite{Schmidt.2019}, originally treated in \cite{Sloan.1974}, the Bose polaron \cite{Lampart.2019} and more general polaron models \cite{Lampart.2023}.
A similar approach is used in the abstract results of \cite{Posilicano.2020,Posilicano.2024, BinzLampart.2021}.
We also employ an IBC method as part of our technique.
Renormalization of the mentioned models has also been carried out using functional integration methods, see \cite{GubinelliHiroshimaLorinczi.2014,MatteMoller.2018} for the Nelson model and \cite{HinrichsMatte.2022,HinrichsMatte.2023} for the 2d semi-relativistic Nelson model and \cite{HinrichsMatte.2024} for the Fr\"ohlich polaron.
Also, let us note that a toy model for interacting quantum fields was recently treated in \cite{AlvarezMoller.2022,AlvarezMoller.2023} using a resummation technique of Eckmann \cite{Eckmann.1970}.

In all of these renormalization results it has also been investigated how spectral properties, especially the infrared behavior, persist under the removal of the ultraviolet cutoff.
For example, the renormalized Nelson model was studied in \cite{BachmannDeckertPizzo.2012,HiroshimaMatte.2019,DamHinrichs.2021}, the semi-relativistic two-dimensional Nelson model in \cite{HinrichsMatte.2022} and the Bose polaron in \cite{HinrichsLampart.2023}, also see references therein.
Further, the dependence of the mass shell of the ultraviolet cutoff or models which do not allow for self-energy renormalization has been studied, e.g., for the three-dimensional semi-relativistic Nelson model in \cite{DeckertPizzo.2014} and the Pauli--Fierz model in \cite{LiebLoss.2005,BachHach.2022}.
Our result thus yields the interesting question, whether the characteristic infrared behavior of the spin boson model decribed above can be recovered for the Hamiltonians constructed in this article.
We leave this problem for future research, especially because the results from \cite{HaslerHinrichsSiebert.2021c,BetzHinrichsKraftPolzer.2025} indicate that a treatment of this question will require a functional integration representation of the renormalized spin boson models, at least in the infrared-divergent case.

\medskip
For the remainder of this introduction, let us discuss generalized spin boson models given on the Hilbert space $\IC^N\otimes \FS$ of the form
\begin{align}\label{def:SB}
	\HSBL \coloneqq A\otimes \Id + \Id\otimes\dG(\omega) + B^*\otimes a(v_\Lambda) + B\otimes\ad(v_\Lambda).
\end{align}
Here, $N\in\IN$ is the number of degrees of freedom of the quantum mechanical system, $A\in\IC^{N\times N}$ is a selfadjoint matrix and $B\in\IC^{N\times N}$ is, for now, an arbitrary matrix.
Further, $\FS$ is the bosonic Fock space over $L^2(\cM,\d\mu)$, where $(\cM,\Sigma,\mu)$ is some $\sigma$-finite measure space, $\dG(\omega)$ is the second quantization of the boson dispersion $\omega:\cM\to[0,\infty)$ with $\omega>0$ almost everywhere and $v:\cM\to\IC$ is the form factor or coupling function such that $v_\Lambda=\chr_{\{\omega<\Lambda\}}v\in L^2(\cM,\d\mu)$ and $\omega^{-1/2}v_\Lambda\in L^2(\cM,\d\mu)$ for arbitrary ultraviolet cutoff $\Lambda>0$.
The precise notation is rigorously introduced in \cref{subsec:Fock}
and we remark that the models treated in our main result are in fact more general than \cref{def:SB}.
Nevertheless, we here stick to the above model for notational simplicity.

 The most important examples are
\begin{enumerate}
	\item[\hypertarget{ex1}{[Ex.\,1]}] the (standard) spin boson model with $N=2$, $A=\sigma_z = \begin{pmatrix}1 & 0\\0 & -1\end{pmatrix}$ and $B = \sigma_x = \begin{pmatrix} 0 & 1 \\ 1 & 0\end{pmatrix}$;
	\item[\hypertarget{ex1}{[Ex.\,2]}] the rotating wave approximation with $N=2$, $A=\sigma_z$, $B=\sigma_- = \begin{pmatrix} 0 & 0 \\ 1 & 0\end{pmatrix}$;
	\item[ {[Ex.\,3]}] their many-spin counterparts, i.e., with $N=2^k$, $B\in\{\sigma_x^{\otimes k},\sigma_-^{\otimes k}\}$ (and $A$ incorporating arbitrary spin-spin interactions).
\end{enumerate}
We are now interested in studying the limit $\Lambda\to\infty$.
Note that the expression \cref{def:SB} with $\Lambda=\infty$ only yields a well-defined selfadjoint operator $H_{\sfS\sfB,\infty}$ if $v\in L^2(\cM,\d\mu)$  and is only bounded from below if $\omega^{-1/2}v\in L^2(\cM,\d\mu)$.
Otherwise, we require an ultraviolet renormalization, in which we subtract the divergent contribution to the ground state energy. 

Our main result \cref{mainthm} concerns the renormalizability in the case $(1+\omega)^{-1}v\in L^2(\cM,\d\mu)$.
Under the assumption that $B$ is either normal or 2-nilpotent, which is the case in all above examples, we prove that $\HSBL+B^*BE_\Lambda$ with $E_\Lambda=\|\omega^{-1/2}v_\Lambda\|^2$ converges to a selfadjoint lower-semibounded operator $\HSB$ in the norm resolvent sense as $\Lambda\to\infty$, i.e.,
\begin{align}\label{eq:HSBconv}
	\big(\HSBL + B^*BE_\Lambda + \ii\big)^{-1} \xrightarrow[\cB(\IC^N\otimes\FS)]{\Lambda\to\infty} \big(\HSB +\ii\big)^{-1}.
\end{align}
In fact, our result also allows $B$ to be a sum of two commuting matrices, one being normal and the other being 2-nilpotent, thus allowing for a large class of interactions.

We also explicitly construct the limiting Hamiltonian $\HSB$ as follows:

In the case that $B$ is normal, let $U$ be the unitary that diagonalizes $B$ and $\lambda_1,\ldots,\lambda_N$ be the corresponding eigenvalues. Then
\begin{align}
	U\HSBL U^* = UAU^* + \operatorname{diag}(\dG(\omega) + \overline{\lambda_i}a(v_\Lambda) + \lambda_i \ad(v_\Lambda))_{i=1,\ldots,N}.
\end{align}
The renormalization energy for this operator is just the diagonal matrix with entries $\abs{\lambda_i}^2\|\omega^{-1/2}v_\Lambda\|$ and the existence of the limit $\Lambda\to\infty$ for each entry  can be deduced using a Weyl operator as dressing transformation.

In the case that $B$ is 2-nilpotent, one can instead rewrite
\begin{align}\label{eq:IBCsimple}
	 \HSBL =
	   A - 1 + (1+ G_\Lambda)\big(1+\dG(\omega) - T_\Lambda\big)(1+G_\Lambda^*) - B^*B\|\omega^{-1/2}v_\Lambda\|_2^2,
\end{align}
where $G_\Lambda = B^*a(v_\Lambda)\big(\dG(\omega)+1)^{-1}$ and
\begin{align}\label{eq:Tsimple}
	T_\Lambda = B^*a(v_\Lambda)(\dG(\omega)+1)^{-1}\ad(v_\Lambda)B - B^*B\|\omega^{-1/2}v_\Lambda\|_2^2.
\end{align}
Then, observing that $G_\infty$ and $T_\infty$ (after appropriate normal ordering) are also well-defined, bounded, respectively $\dG(\omega)$-bounded operators, is the key observation for the renormalization in this case.
We remark that the renormalization term $B^*B\|\omega^{-1/2}v_\Lambda\|_2^2$ in the nilpotent case only adds the renormalization energy $\|\omega^{-1/2}v_\Lambda\|_2^2$ to the excited sectors of the model, e.g., in the case $B=\sigma_-$ where $B^*B = \begin{psmallmatrix}1 & 0 \\ 0 & 0\end{psmallmatrix}$.

Finally, at least in the diagonalizable case, we prove optimality of our renormalization result, by showing that $(1+\omega)^{-1}v\notin L^2(\cM,\d\mu)$ implies that $\HSBL+E_\Lambda$ can not converge in the strong resolvent sense for a large class of reasonable renormalization operators $E_\Lambda$.
The precise result can be found in \cref{thm:nonren}.

\medskip
Let us compare our main results to the literature.

It was already shown in \cite[Corollary 4.4]{DamMoller.2018b} for the standard spin boson model \hyperlink{ex1}{[Ex.\,1]}
that convergence of $\HSBL + \inf\sigma(\HSBL)$ in the strong resolvent sense holds if and only if $(1+\omega)^{-1}v\in L^2(\cM,\mu)$.
In this case the limiting operator was explicitly constructed as well, employing a different diagonalization of the model, due to the spin parity symmetry of the model.
Our result for the standard spin boson model thus extends theirs to norm resolvent convergence and our simpler construction of the renormalized Hamiltonian does not rely on symmetry assumptions.
Furthermore, we extend the proof that $(1+\omega)^{-1}v\in L^2(\cM,\d\mu)$ is a necessary assumption for self-energy renormalizability to arbitrary normal matrices $B$.

Lonigro \cite{Lonigro.2022,Lonigro.2023} treated general massive, i.e., $\omega\ge 1$, spin boson models with mild ultraviolet divergences, i.e., $(1+\omega)^{-1/2}v\in L^2(\cM,\d\mu)$, in which no self-energy renormalization is necessary. He also proved renormalizability of the rotating wave spin boson model \hyperlink{ex2}{[Ex.\,2]} in the cases $(1+\omega)^{-s}v\in L^2(\cM,\d\mu)$ for some $s<1$ and $(1+\omega)^{-1}v\in L^2(\cM,\d\mu)$ for small coupling, i.e., small $\|(1+\omega)^{-1}v\|_2$. 
The method uses a Hilbert scales technique, is also constructive and yields convergence of the regularized models in strong resolvent sense.
We extend the results of Lonigro to arbitrary large coupling, the massless case, and provide norm resolvent convergence.
Recently, Lill and Lonigro \cite{LillLonigro.2024} also applied a result of Posilicano \cite{Posilicano.2020} to spin boson models, which is closely related to the IBC method we use for 2-nilpotent interactions, also treating the mild ultraviolet divergences discussed above (but norm resolvent convergence).

\subsection*{Acknowledgments}
BH and JVM thank Sascha Lill and Davide Lonigro for valuable discussion on the ultraviolet problem for the spin boson model. 
BH and JVM acknowledge support by the Ministry of Culture and Science of the State of North Rhine-Westphalia within the project `PhoQC' (Grant Nr. PROFILNRW-2020-067).
JL was supported by the Agence National de la Recherche (ANR, project MaBoP ANR-23-CE40-0025) and the EIPHI Graduate School (ANR-17-EURE-0002) together with the Bourgogne-Franche-Comté Region through the project SQC.
\subsection*{Structure of the Article}
The remainder of this article is structured as follows. In \cref{sec:modelandresults}, we introduce the precise class of models considered, which is a generalization of \cref{def:SB}, and state our main renormalization result \cref{mainthm} as well as the non-renormalizability result \cref{thm:nonren}.
In \cref{sec:IBC}, we then set up the IBC approach for 2-nilpotent interactions.
Normal interactions, the corresponding dressing transformations and the proofs of our main results are then given in \cref{sec:normal}.

\section{Model and Results}\label{sec:modelandresults}
Throughout this article, we will fix two central objects, the first being the boson momentum space $\cM$.
\begin{hyp}Throughout this article, we will work under the standing assumptions that
	\begin{enumhyp}
		\item $(\cM,\mu)$ is a $\sigma$-finite measure space,
		\item $\omega:\cM\to[0,\infty)$ is measurable, $\omega>0$ almost everywhere.
	\end{enumhyp}
\end{hyp}
To distinguish between infrared and ultraviolet sectors, we introduce a fixed infrared cutoff $\kappa>0$ as well as the notation
\begin{align*}
	f_{\le} \coloneqq \chr_{\{\omega\le \kappa\}}f,
	\qquad
	f_> \coloneqq \chr_{\{\omega> \kappa\}} f
\end{align*}
for arbitrary vector space valued functions $f$ defined on $\cM$.
See \cref{rem:IRcutoff} for a discussion of the several roles of $\kappa$ in our main result.

Let us also fix the spin space $\HS_\sfs$, which we (different from the introduction) do not require to be finite-dimensional.
\begin{hyp}
	Throughout this article, we will work under the standing assumptions that $\HS_\sfs$ is an arbitrary fixed Hilbert space.
\end{hyp}

\subsection{Fock Space Calculus}
\label{subsec:Fock}

Let us briefly introduce the relevant notation for bosonic Fock spaces needed here, for more details see for example \cite{Parthasarathy.1992,Arai.2018}.

Setting $\fh\coloneqq L^2(\cM,\mu)$, we work on the Hilbert space
\begin{align}\label{def:Fock}
	\HS \coloneqq \HS_\sfs \otimes \FS(\fh) \cong \HS_\sfs \oplus \bigoplus_{n\in\IN} L^2_\sfs(\cM^n,\d \mu^{\otimes n};\HS_\sfs),
\end{align}
where $\mu^{\otimes n}$ denotes the $n$-fold product measure of $\mu$ with itself (recall that we assumed $(\cM,\mu)$ to be $\sigma$-finite) and we denote by $L^2_\sfs$ the Bochner--Lebesgue space of symmetrized squareintegrable $\HS_\sfs$-valued functions, where we use the symmetrization
\begin{align*}
	&L^2_\sfs(\cM^n,\mu^{\otimes n};\HS_\sfs) \\&\qquad \coloneqq \big\{ f\in L^2(\cM^n,\d \mu^{\otimes n};\HS_\sfs) \big| f(k_{\pi(1)},\ldots,k_{\pi(n)}) = f(k_1,\ldots,k_n),\ \text{$\mu$-a.e. $k\in\cM^n$},\ \pi\in\fS_n \big\}.
\end{align*}
In \cref{def:Fock}, $\FS(\hs)\coloneqq\bigoplus_{n\in\IN}L^2_\sfs(\cM^n,\d \mu^{\otimes n})$ is the usual bosonic Fock space, which we will not further utilize in this article.
We will also utilize the finite boson subspace
\begin{align*}
	\HS_{\sff\sfi\sfn} \coloneqq \HS_\sfs \oplus \widehat\bigoplus_{n\in\IN} L^2_\sfs(\cM^n,\d \mu^{\otimes n};\HS_\sfs),
\end{align*}
 where $\widehat\oplus$ denotes the algebraic direct sum, i.e., the finite sequences.

The free energy of the bosonic field will now be given by the second quantization operator $\dG(\omega)$ given as the direct sum of multiplication operators
\begin{align}
	\dG(\omega) \coloneqq \bigoplus_{n\in\IN_0} \dG^{(n)}(\omega),
	\quad
	\dG^{(n)}(\omega)(k_1,\ldots,k_n) \coloneqq \sum_{i=1}^n\omega(k_i)\psi^{(n)}(k),
	\quad
	\dG^{(0)}(\omega) = 0.
\end{align}
Interactions between spin and field will be given by square-integrable functions taking values in the bounded operators on the spin Hilbert space, whence we abbreviate $\fb_s \coloneqq L^2(\cM,\omega^{-s}\d\mu;\cB(\HS_\sfs))$, again in the sense of a Bochner--Lebesgue space.

	\begin{rem}
		\label{rem:bsinclusion}
		It is worth noting some inclusion relations between the $\fb_s$ spaces. More precisely, if $F=F_\le$ and $F\in\fb_s$, then $F\in \fb_{s'}$ for all $s'\leq s$. Conversely, if $F=F_>$, then $F\in\fb_s$ implies $F\in\fb_{s'}$ for every $s'\geq s$. This is, if we denote
		\begin{align*}
		\fb_{s}^{\#}:=\fb_s\cap\{F:\mathcal{M}\longrightarrow \mathcal{B}(\mathcal{H}_s)\,\,\text{with}\,\, F=F_{\#}\}\,,\quad \#\in\{\leq,>\}
		\end{align*}
		then
		\begin{align*}
			\cdots\subset \fb^\leq_{1}\subset \cdots \subset \fb^\leq_{0}\subset\cdots\subset \fb^\leq_{-1}\subset\cdots\qquad\text{and}\qquad \cdots\subset \fb_{-1}^>\subset \cdots \subset \fb_{0}^>\subset\cdots\subset \fb_{1}^>\subset\cdots.
		\end{align*}
	\end{rem}
We will further write
\begin{align}
	\label{eq:bs}
	\braket{F,G}_{\fb_s} \coloneqq \int_{\cM} F(k)^*G(k)\omega^{-s}\d\mu(k) \in \cB(\HS_\sfs), \quad F,G\in\fb_s,
\end{align}
and for arbitrary (strongly) measurable $F,G:\cM\to\cB(\HS_\sfs)$ define the pointwise commutator $[F,G]:\cM\times\cM\to\cB(\HS_\sfs)$ by
\begin{align*}
	[F,G](k,p) \coloneqq [F(k),G(p)] \coloneqq F(k)G(p)-G(p)F(k),\qquad k,p\in\cM.
\end{align*}
Then,
given $F\in \fb_0$ and $\psi=(\psi^{(n)})_{n\in\IN_0}\in\HS_{\sff\sfi\sfn}$, we define the annihilation operator
\begin{align*}
		&(a(F)\psi)^{(n)}(k_1,\ldots,k_n) \coloneqq \sqrt{n+1}\int_\cM F(k)^*\psi^{(n+1)}(k,k_1,\ldots,k_n)\d\mu(k),
\end{align*}
which is clearly well-defined as a Bochner--Lebesgue integral.
Its adjoint $\ad(F)\coloneqq a(F)^*$, which exists since $\HS_{\sff\sfi\sfn}$ is dense in $\HS$, acts as
\begin{align*}
	(\ad(F)\psi)^{(n)}(k_1,\ldots,k_n) = \frac{1}{\sqrt{n}}\sum_{i=1}^{n}F(k_i)\psi^{(n-1)}(k_1,\ldots,\cancel{k_i},\ldots,k_n).
\end{align*}
As usual, we will also formally set $F=\delta(k-\cdot)$ and obtain the pointwise annihilation and creation operators $a_k$ and $\ad_k$, respectively, in the sense of operator valued distributions.

Furthermore, if $F\in\fb_0$, the operator $\ph(F) \coloneqq  \overline{a(F)+\ad(F)}$ (where $\overline{\cdot}$ denotes the operator closure) is selfadjoint,
as one can readily check using Nelson's analytic vector theorem and the estimate \cref{lem:simple.relbound} below with $\omega=1$, cf. \cite[Theorem 5.22]{Arai.2018} for details.

Let us collect some well-known properties of the operators introduced up to now, for which proofs can be found in any of the referenced textbooks (usually for the case $\HS_\sfs=\IC$, but all proofs can be carried out word-by-word in our setting).
\begin{lem}\label{lem:simple}
	Assume $F,G\in\fb_0$.
	\begin{enumlem}
		\item\label{lem:simple.relbound} If $F\in\fb_1$, then $\sD(\ph(F))\subset \sD(\dG(\omega)^{1/2})$ and
		\[ \|\ph(F)\psi\| \le 2\|(1+\omega^{-1/2})F\|_{\fb_0}\|(1+\dG(\omega))^{1/2}\psi\|.\]
		\item\label{lem:simple.CCRph} If $[F,G]=[F,G^*]=0$, then $[\ph(F),\ph(G)] = \braket{F,G}_{\fb_0} - \braket{G,F}_{\fb_0}$ holds on $\HS_{\sff\sfi\sfn}$.
		\item\label{lem:simple.CCRdG} If $F\in \fb_{-2}$, then $[\dG(\omega),\ph(F)] = \ii\ph(\ii\omega F)$ holds on $\HS_{\sff\sfi\sfn}$. 
	\end{enumlem}
\end{lem}

\subsection{UV-Regular GSB Hamiltonian}\label{subsec:UVreg}
We can now define the generalized spin boson (GSB) Hamiltonian for ultraviolet regular interactions.

Given $S\in\cB(\HS)$ symmetric, and hence selfadjoint, and $V\in \fb_0\cap \fb_{1}$, the Hamiltonian is now given by
\begin{align}\label{def:Hreg}
	\Hreg(S,V) \coloneqq S + \dG(\omega) + \ph(V).
\end{align}
Applying \cref{lem:simple.relbound} and boundedness of $S$, one easily checks that under these assumptions $\Hreg(S,V)$
is selfadjoint and lower-semibounded on the domain $\sD(\dG(\omega))$, by the Kato--Rellich theorem.
Note that this reproduces the Hamiltonian $\HSB$ from the introduction when $\mathcal{H}_s=\C^N$ is finite-dimensional and $V(k)=v(k) B$ is a scalar multiple of a fixed operator $B\in \mathcal{B}(\mathcal{H}_s)$.

\subsection{Renormalization of UV-Critical GSB Models}
\label{subsec:renormalization}
We want to study the limit of interactions which do not satisfy the assumptions of the previous \cref{subsec:UVreg}.
To this end, we will split up the singular interaction $V$ into three commuting parts: the infrared one $V_\le$, the normal one $V_\sfD$ and the 2-nilpotent one $V_\sfN$.

As described in \cref{eq:IBCsimple},
we will treat the 2-nilpotent part with the interior boundary condition method from \cite{LampartSchmidt.2019}.
Thus, for $F\in\fb_s$, $s\in[1,2]$ and $\lambda>0$, we define the $\dG(\omega)$-bounded operators 
\begin{align}\label{def:T}
	\begin{aligned}
		\Theta_{0}(F,\lambda) &\coloneqq \int_{\cM} \left(F(k)^*\omega(k)^{-1}(\dG(\omega)+\lambda)(\dG(\omega)+\omega(k)+\lambda)^{-1}F(k)\right)\d\mu(k),\\
		\Theta_{1}(F,\lambda) &\coloneqq \int_{\cM^2} \ad_q F(r)^*(\dG(\omega)+\omega(r)+\omega(q)+\lambda)^{-1}F(q)a_r\d\mu(q) \d\mu (r),\\
		T_{F,\lambda} &\coloneqq \Theta_{0}(F,\lambda) + \Theta_1(F,\lambda)
	\end{aligned}
\end{align}
see \cref{prop:relbound} for a proof of the relative boundedness. Further, we define the bounded operator 
\begin{align}\label{def:B}
	G_{F,\lambda} \coloneqq \int_{\cM} F(k)^*a_k(\dG(\omega)+\lambda)^{-1}\d\mu(k),
\end{align}
see \cref{prop:Bbound} for details of the simple proof of boundedness.
We note that $G_{F,\lambda} = a(F)(\dG(\omega)+\lambda)^{-1}$ if $F\in\fb_0\cap\fb_s$.

For $\lambda>0$ and either $s<2$ or $\norm{F}_{\fb_2}<\frac12$, we can then define the selfadjoint operator
\begin{align}\label{def:HIBCnew}
	\begin{aligned}
		& \sD(\HIBCl(F)) \coloneqq (1+G_{F,\lambda}^*)^{-1}\sD(\dG(\omega)),\\
			&\HIBCl(F) \coloneqq (1+G_{F,\lambda})\big(\dG(\omega)+\lambda-T_{F,\lambda}\big)(1+G_{F,\lambda}^*),
	\end{aligned}
\end{align}
cf. \cref{prop:HIBC}.

Our motivation to introduce these operators comes from the following generalization of \cref{eq:IBCsimple},
which we will prove in the end of \cref{subsec:nilreg}.
\begin{prop}\label{thm:nilpotent}
	Let $V_\sfN\in\fb_0\cap\fb_1$ satisfy $V_\sfN(k)V_\sfN(p)=0$ for $\mu$-almost all $k,p\in\cM$.
	Then
	\begin{align*}
		\Hreg(0,V_\sfN) = \HIBCl(V_\sfN) - \braket{V_\sfN,V_\sfN}_{\fb_1} - \lambda \qquad \mbox{for all}\ \lambda\ \mbox{sufficiently large}.
	\end{align*}
\end{prop}

The normal part of the interaction will be treated by a dressing transformation, similar to ideas already going back to van Hove \cite{vanHove.1952}.
The dressing transformation acts on the full Hilbert space and is given by the unitary
\begin{align}\label{def:Weyl}
	W(F) \coloneqq \e^{\ii\ph(\ii F)},\quad F\in\fb_0
\end{align}
acting on the full Hilbert space $\HS$. It is
 well-defined by the selfadjointness of the field operator discussed above and we will refer to it as {\em generalized Weyl operators} here.
\begin{rem}
	Let us briefly compare our generalized Weyl operators with the usual ones defined in the case $\HS_\sfs = \IC$, see for example \cite{Parthasarathy.1992}. Assume $B\in\cB(\HS_\sfs)$ is a normal operator with corresponding spectral measure $\sfE_B$. Then for $f\in L^2(\cM,\mu)$, we have
	\[ W(Bf) = \int_{\IC} W(\lambda f)\d \sfE_B(\lambda), \]
	where $W(\lambda f)$ now is the usual Weyl operator acting on $\FS(\hs)$, cf. \cref{def:Fock}.
	Note however that the properties of $W(F)$ will be different from those of the standard Weyl operators when $[F,F]\neq 0$ outside a $\mu^{\otimes 2}$-zero set.
\end{rem}
We can now state our main result.
\begin{thm}
	\label{mainthm}
 Let $S\in\cB(\HS)$ be symmetric and let $V=V_\le + V_\sfD+V_\sfN:\cM\to\cB(\HS_\sfs)$ be measurable such that $[V_\#,V_\lozenge]=0$ for $\#,\lozenge\in\{\sfD,\sfN\}$, $V_\sfD(p)$ is normal and  $V_\sfN(k)V_\sfN(p)=0$ for $\mu$-almost all $k,p\in\cM$.

 Assume that $V_\le\in\fb_1$, $V_\sfD\in\fb_2$ and
	$V_\sfN\in\fb_{s_\sfN}$ for some $s_\sfN\in[1,2]$.
	If $s_\sfN<2$ or if $\kappa>0$ is large enough that $\norm{V_\sfN}_{\fb_2}<\frac12$,
	then the operator $H(S,V)$ given by
	\begin{align*}
		\begin{aligned}
			&\sD(H(S,V)) \coloneqq W(\omega^{-1}V_\sfD)\sD(\HIBCl(V_\sfN))\supset W(\omega^{-1}V_\sfD)\sD(\dG(\omega_\le))\\
			&H(S,V) \coloneqq S + W(\omega^{-1}V_\sfD)\big(\HIBCl(V_\sfN) +  \ph(V_\le)\big)W(\omega^{-1}V_\sfD)^* - \lambda
		\end{aligned}
	\end{align*}
	is selfadjoint, lower-semibounded and independent of $\lambda>0$.
	
	Further, for any sequence $(V_n=V_{n,\le} + V_{n,\sfD} + V_{n,\sfN})_{n\in\IN}\subset\fb_0\cap\fb_1$ satisfying the assumptions from above as well as $V_{n,\le}\xrightarrow{\fb_1}V_\le$ and $V_{n,\#}\xrightarrow{\fb_2}V_\#$ for $\#\in\{\sfD,\sfN\}$, the operators
	\begin{align*}
		\Hreg(S,V_n) + \braket{V_{n,>},V_{n,>}}_{\fb_1}
	\end{align*}
	converge to $H(S,V)$ in the strong resolvent sense. The convergence is in norm resolvent sense, if either $V_{n,\sfD}=0$ for all $n\in\IN$ or $s_\sfN<2$, $V_{n,\sfN}\xrightarrow{\fb_{s_\sfN}}V_\sfN$.
\end{thm}
\begin{rem}\label{rem:IRcutoff}
	Let us comment on the assumption $\|V_\sfN\|_{\fb_2}<\frac12$ in the critical case that no $s_\sfN<2$ with $V_\sfN\in\fb_{s_\sfN}$ exists.
	Since $V_\sfN\in\fb_2$ and $V_\sfN=V_{\sfN,>}$, we observe that $\chr_{\{\omega\le\sigma\}}V_\sfN\in\fb_1$ for arbitrary $\sigma\ge\kappa$, cf. \cref{rem:bsinclusion}. Hence, we can always choose $\kappa$ large enough that $\|V_\sfN\|_{\fb_2}<\frac12$, since we do not require any structural assumptions on $V_\le\in\fb_1$.
	Thus, the assumption $\|V_\sfN\|_{\fb_2}$ is not a restriction on the coupling constant but merely an incorporation of the infrared cutoff $\kappa$  in the exact construction of $H(S,V)$.
\end{rem}
\begin{rem}\label{rem:renenergy}
	The expression $\braket{V_{n,>},V_{n,>}}_{\fb_1}$ can be considered as the divergent renormalization energy, required to control the ultraviolet divergences of $\Hreg(S,V_n)$ in the limit $n\to\infty$.
	Since $V_\le\in\fb_1$, we remark that it only differs from $\braket{V_n,V_n}_{\fb_1}$ by operators on $\HS_\sfs$ uniformly bounded in $n$, whence we have also proven the convergence of $\Hreg(S,V_n) + \braket{V_n,V_n}_{\fb_1}$ as claimed in  \cref{eq:HSBconv}.
\end{rem}
The proof of \cref{mainthm} will be given in the end of \cref{subsec:dressren}, combining results from \cref{sec:IBC,sec:normal}.
\subsection{Non-Renormalizability in the Super-Critical Case}
Let us complement the renormalization result from \cref{subsec:renormalization} by the following non-renormalizability statement, which is a generalization of the equivalence statement \cite[Corollary 4.4]{DamMoller.2018b}. Also note that a similar result for the three-dimensional semi-relativistic Nelson model was proven in \cite{DeckertPizzo.2014}.
\begin{thm}\label{thm:nonren}
	Let $V=V_\leq+V_\sfD:\cM\to\cB(\cH_s)$ be measurable, with $V_\le\in\fb_1$.
	Assume that there exist $\psi\in\mathcal{H}_s$ and $v:\cM\to\IC$ measurable such that
	\begin{align*}
		V_\sfD(k)\psi=v(k)\psi\,\qquad\mbox{$\mu$-a.e.}\ k\in\cM.
	\end{align*} 
	If $v\omega^{-1}\notin L^2(\cM)$ , then for any sequence of operators $(V_{\sfD,n})\subset \fb_0^>$ such that 
	\begin{align*}
		V_{\sfD,n}(k)\longrightarrow V_D(k)\,\qquad \text{$\mu$-a.e.}\ k\in\cM\qquad\text{with}\quad V_{\sfD,n}(k)\psi=v_{n}(k)\psi
	\end{align*}
	and every sequence $(B_n)\subset \cB(\HS_\sfs)$ of symmetric operators such that \[\liminf_{n\to\infty} \big(\braket{\psi,B_n\psi}+\inf\sigma(\Hreg(S,V_{\sfD,n}))\big)>-\infty,\]
	the operator $\Hreg(S,V_\leq+V_{\sfD,n})+B_n$ can not converge in the strong resolvent sense. 
\end{thm}
\begin{rem}
	We have assumed that $\psi$ is also an eigenvector of the approximating sequence. This seems to be a reasonable assumption and is, for example, satisfied for any cut-off approximation of the form $V_{\sfD,n}(k) = \chi_n(k)V_\sfD(k)$.
	Furthermore, the lower bound on $\braket{\psi,B_n\psi}$ essentially states that the ground state energy of the approximating Hamiltonians after addition of the correction $B_n$ should be uniformly bounded from below. This reflects the fact that the studied renormalization schemes should correspond to an appropriate self-energy renormalization.
\end{rem}
The proof of \cref{thm:nonren} is presented in \cref{subsec:nonren}.

\section{IBC for 2-Nilpotent Interactions}\label{sec:IBC}
In this \lcnamecref{sec:IBC}, we consider the operators defined in \cref{def:T,def:B}. In \cref{subsec:relbound}, we provide the relevant (relative) bounds and prove continuity w.r.t the form factors, which will be important for approximation by regularized models.
In \cref{subsec:nilreg}, we then show how they relate to \cref{eq:Tsimple} by normal-ordering and prove \cref{thm:nilpotent}.
\subsection{(Relative) Bounds and Continuity of \texorpdfstring{$T_{F,\lambda}$}{T} and \texorpdfstring{$G_{F,\lambda}$}{G}}
\label{subsec:relbound}
Let us start, by considering the two summands of $T_{F,\lambda}$.
\begin{prop}\label{prop:relbound}
	Let $\ell\in\{0,1\}$ and $s\in[1,2]$. Then, for all $F_1,F_2\in\fb_{s}$ and $\lambda>0$,  we have  that $\sD(\Theta_{\ell}(F_1,\lambda))\cap\sD(\Theta_{\ell}(F_2,\lambda))\supset \sD(\dG(\omega)^s)$ and for $\psi\in\sD(\dG(\omega)^s)$
	\begin{align*}
		&\big\|\big(\Theta_{\ell}(F_1,\lambda)-\Theta_{\ell}(F_2,\lambda)\big)\psi\big\| \le (\|F_1\|_{\fb_s}+\|F_2\|_{\fb_s}) \|F_1-F_2\|_{\fb_s}\|(\dG(\omega)+\lambda)^{ s -1}\psi\|.
	\end{align*}
\end{prop}
\begin{proof}
	For notational convenience, we only treat the case $F_2=0$, but the extension to arbitrary $F_2$ is obvious.
	
	Let us first consider the case $\ell=0$.
	The corresponding $\cB(\HS_\sfs)$-valued integral kernel is
	\begin{align*}
		\theta_{0,\lambda}(E,k) = F_1(k)^*\frac{E}{(E+\omega(k)+\lambda)\omega(k)}F_1(k).
	\end{align*}
	Integrating with respect to $k$ and applying the triangle inequality, we find
	\begin{align*}
		\left\|\int_{\cM}\theta_{0,\lambda}(E,k)\d\mu(k)\right\|_{\cB(\HS_\sfs)}
		&\le
		\int_{\cM} \frac{\|F_1(k)\|_{\cB(\HS_\sfs)}^2}{(E+\omega(k)+\lambda)^{s-1}\omega(k)} \frac{E}{(E+\omega(k)+\lambda)^{2-s}} \d\mu(k)
		\\
		&\le
		\|F_1\|_{\fb_s}^2(E+\lambda)^{s-2}E.
	\end{align*}
	Inserting $E=\dG(\omega)$ and applying the spectral theorem yields the statement for $\ell=0$.
	
	We move to the case $\ell=1$.
	Again we can first consider the integral kernel
	\begin{align*}
		\theta_{1,\lambda}(E,p,q) & = {F_1}(p)^*(E+\omega(p)+\omega(q)+\lambda)^{-1}F_1(q)
	\end{align*}
	and estimate
	\begin{align*}
		\|\theta_{1,\lambda}(E,p,q)\|_{\cB(\HS_\sfs)} \le (E+\omega(p)+\lambda)^{-s/2}\|F_1(p)\|_{\cB(\HS_\sfs)}(E+\omega(q)+\lambda)^{s/2-1}\|F_1(q)\|_{\cB(\HS_\sfs)}.
	\end{align*}
	Inserting this into \cref{def:T}, using that $F_1$ and $\dG(\omega)$ commute and twice applying the Cauchy--Schwarz inequality,
	we find
	\begin{align}
		\abs{\braket{\Phi,\Theta_{1,\lambda}\Psi}}
		&= 
		\int_{\cM^2} \abs{\braket{a_p\Phi,\theta_{1,\lambda}(\dG(\omega),p,q)a_q\Psi}}\d\mu(p)\d\mu(q)
		\nonumber\\
		&   
		\label{eq:productestimate}
		\begin{aligned}
			&\le\int_{\cM} \norm{F_1(p)}_{\cB(\HS_\sfs)}\norm{(\dG(\omega)+\omega(p)+\lambda)^{-s/2}a_p\Phi}\d\mu(p)
			\\&\qquad \times
			\int_{\cM} \norm{F_1(q)}_{\cB(\HS_\sfs)}\norm{(\dG(\omega)+\omega(q)+\lambda)^{s/2-1}a_q\Psi}\d\mu(q)
		\end{aligned}
	\end{align}
	Now note that the Cauchy--Schwarz inequality also yields
	\begin{align}\label{eq:a-sbound}
		\int_{\cM} \norm{F_1(p)}_{\cB(\HS_\sfs)}  \norm{a_p\Psi}\d\mu(p)
		\le \norm{F_1}_{\fb_s} \Big(\int_{\cM} \omega(p)^s\norm{a_p\Psi}^2 \d p\Big)^{1/2} \le \norm{F_1}_{\fb_s}\norm{\dG(\omega)^{s/2}\Psi}
	\end{align}
	Combining \cref{eq:productestimate,eq:a-sbound} as well as  the pull-through formula \[(\dG(\omega)+\omega(p)+\lambda)^{\alpha}a_p\Phi = a_p(\dG(\omega)+\lambda)^{\alpha}\Phi, \qquad \alpha\in\IR,\]
	we find
	\begin{align*}
		\abs{\braket{\Phi,\Theta_{1,\lambda}\Psi}} \le \|F_1\|_{\fb_s}^2\|(\dG(\omega)+\lambda)^{s-1}\Psi\|\|\Phi\|, \quad \Psi \in \sD(\dG(\omega)^{s-1}),\quad \Phi\in\HS.
	\end{align*}
	Since $\Phi$ is arbitrary, this proves the statement.
\end{proof}
This allows us to treat $T_{F,\lambda}$ as a small perturbation of $\dG(\omega)+\lambda$.
	\begin{cor}\label{cor:ProvCorPositvity}
	Let $F\in\fb_s$ for some $s\le2$. If $s<2$ or $\norm{F}_{\fb_2}<\tfrac12$, then the operator $\dG(\omega)-T_{F,\lambda}+\lambda$ with domain $\sD(\dG(\omega))$  is selfadjoint and lower-semibounded  for all $\lambda>0$.
	\end{cor}
	\begin{proof}
		\cref{prop:relbound,def:T} give the estimate
		\begin{align*}
			\|T_{F,\lambda}\psi\|\leq 2\|F\|_{\fb_s}^2\|(d\Gamma(\omega)+\lambda)^{s-2}(d\Gamma(\omega)+\lambda)\psi\|\qquad \text{for every }\,\psi\in\sD(d\Gamma(\omega))\,.
		\end{align*}
		The statement in the case $\norm{F}_{\fb_2}<\tfrac12$ thus directly follows from the Kato--Rellich theorem.
		For $s<2$, we additionally employ that $(\dG(\omega)+\lambda)^{s-1}$ is infinitesimally $\dG(\omega)$-bounded.
	\end{proof}
We apply the same procedure as in \cref{prop:relbound} to $G_{F,\lambda}$.
\begin{prop}\label{prop:Bbound}
	Let $s\in[1,2]$ and $F\in\fb_s$. Then, for all $\lambda>0$, we have
	\begin{align*}
		\|G_{F,\lambda}\| \le \|F\|_{\fb_s}\lambda^{(s-2)/2}.
	\end{align*}
\end{prop}
\begin{proof}
	Using the definition \cref{def:B,eq:a-sbound}, we immediately find
	\[
		\norm{G_{F,\lambda}} \le \norm{F}_{\fb_s}\norm{\dG(\omega)^{s/2}(\dG(\omega)+\lambda)^{-1}} \le \|F\|_{\fb_s}\lambda^{(s-2)/2}.\qedhere
	\]
\end{proof}
To prove selfadjointness of the operator defined in \cref{def:HIBCnew}, by \cref{lem:adjointtrans}, we will need $1+G_{F,\lambda}$ to have bounded inverse.
By the Neumann series, the above \lcnamecref{prop:Bbound} provides that $1+G_{F,\lambda}$ has bounded inverse if the right hand side of the upper bound is strictly smaller than 1.
This is the usual argument applied in the IBC approach.

Since we only consider 2-nilpotent interactions here, we can extend this assumption by a simple algebraic argument.
\begin{lem}\label{lem:Gboundedinv}
	For all $F\in\fb_2$ such that $F(k)F(p)=0$ for $\mu$-almost all $k,p\in\cM$, the operators $1+G_{F,\lambda}$ and $1+G_{F,\lambda}^*$ have bounded inverse and
	\begin{align*}
		(1+G_{F,\lambda})^{-1} = 1-G_{F,\lambda}, \qquad (1+G_{F,\lambda}^*)^{-1} = 1-G_{F,\lambda}^*.
	\end{align*}
\end{lem}
\begin{proof}
	The nilpotency assumption and the definition \cref{def:B} directly imply $G_{F,\lambda}^2 = 0$ and thus the first statement. The second follows by taking adjoints.
\end{proof}
	Let us summarize the obtained selfadjointness statement
	and  complement it with a continuity result.
	\begingroup
	\renewcommand{\HIBCl}{\Xi_\lambda}
	\begin{prop}
		\label{prop:HIBC}
		Let $F,V\in \fb_2$ and either $\|V\|_{\fb_s}<\infty$ for some $s<2$ or $\|V\|_{\fb_2}<\frac{1}{2}$. Then, for all $\lambda>0$ such that $1+G_{F,\lambda}$ has bounded inverse, the operator 
		\begin{align}
			\label{def:HIBC}
			\HIBCl(F,V)\coloneqq (1+G_{F,\lambda})(\dG(\omega) + \lambda - T_{V,\lambda})(1+G_{F,\lambda}^*)
		\end{align}
		is selfadjoint and lower-semibounded. Further, if $V_n\xrightarrow{\fb_2}V$ and $F_n\xrightarrow{\fb_2}F$ such that $G_{F_n,\lambda}$ has bounded inverse for all $n\in\IN$, then
	$\HIBCl(F_n,V_n)$
	converges to $\HIBCl(F,V)$ in the norm resolvent sense. 
	\end{prop}
	\begin{proof}
		Selfadjointness and lower-semiboundedness of
		$\HIBCl(F,V)$ immediately follow from \cref{cor:ProvCorPositvity,lem:adjointtrans}.
		
		In order to prove the convergence statement,
		we decompose
		\begin{align}
			\label{eq:resdecomp}
			\begin{aligned}
				&\left(\HIBCl(F_n,V_n)+\ii\right)^{-1}-\left(\HIBCl(F,V)+\ii\right)^{-1}
				\\&= 
				\Big(\left(\HIBCl(F,V_n)+\ii\right)^{-1}-\left(\HIBCl(F,V)+\ii\right)^{-1}\Big)
				 + \Big(\left(\HIBCl(F_n,V_n)+\ii\right)^{-1}-\left(\HIBCl(F,V_n)+\ii\right)^{-1}\Big).
			\end{aligned}
		\end{align}
		We rewrite the first summand on the right hand side using the resolvent identity as
		\begin{align*}
			\begin{aligned}
				&\left(\HIBCl(F,V_n)+\ii\right)^{-1}-\left(\HIBCl(F,V)+\ii\right)^{-1} 
				\\&=
				\left(\HIBCl(F,V_n)+\ii\right)^{-1}(1+G_{F,\lambda})(T_{V,\lambda}-T_{V_n,\lambda})(1+G_{F,\lambda}^*)\left(\HIBCl(F,V)+\ii\right)^{-1}
				\\&
					= \left(\HIBCl(F,V_n)+\ii\right)^{-1}(1+G_{F,\lambda})(T_{V,\lambda}-T_{V_n,\lambda})(\dG(\omega)+\lambda)^{-1}(\dG(\omega)+\lambda)(1+G_{F,\lambda}^*)\left(\HIBCl(F,V)+\ii\right)^{-1}\,.
			\end{aligned}
		\end{align*}
		It is easy to see that $(\dG(\omega)+\lambda)(1+G_{F,\lambda}^*)\left(\HIBCl(F,V)+\ii\right)^{-1}$ is a closed operator defined on all of $\HS$ and therefore, by the closed graph theorem, bounded. Furthermore, by \cref{prop:relbound}
		\[\|(T_{V,\lambda}-T_{V_n,\lambda})(\dG(\omega)+\lambda)^{-1}\|\xrightarrow{n\to\infty} 0.\]
		As $\|\left(\HIBCl(F,V_n)+\ii\right)^{-1}(1+G_{F,\lambda})\|$ is uniformly bounded in $n$,
		this gives
		\begin{align}
			\label{eq:conv.1}
			\Norm{\left(\HIBCl(F,V_n)+\ii\right)^{-1}-\left(\HIBCl(F,V)+\ii\right)^{-1}} \xrightarrow{n\to\infty} 0.
		\end{align}
		To treat the second summand on the right hand side of \cref{eq:resdecomp},
		we first observe that by \cref{prop:Bbound} and $G_{F_n,\lambda}-G_{F,\lambda} = G_{F_n-F,\lambda}$
		the operator $1+G_{F_n,\lambda}$ has bounded inverse for $n$ large enough as well.
		Hence, we can apply the
		resolvent identity \cref{lem:residentity} and find
			\begin{align*}
				&\left(\HIBCl(F_n,V_n)+\ii\right)^{-1}-\left(\HIBCl(F,V_n)+\ii\right)^{-1}
				\\
				& =  \left(\HIBCl(F_n,V_n)+\ii\right)^{-1}(G_{F,\lambda}-G_{F_n,\lambda})\HIBCl(0,V_n)(1+G_{F,\lambda}^*)\left(\HIBCl(F,V_n)+\ii\right)^{-1}
				 \\& \quad
				  + \Big(\HIBCl(0,V_n)(1+G_{F_n,\lambda}^*)\left(\HIBCl(F_n,V_n)-\ii\right)^{-1}\Big)^*(G_{F,\lambda}^*-G_{F_n,\lambda}^*)\left(\HIBCl(F,V_n)+\ii\right)^{-1}.
			\end{align*}
		Now employing the bound
		\[\Norm{\HIBCl(0,V_n)(1+G_{F,\lambda}^*)\left(\HIBCl(F,V_n)\pm\ii\right)^{-1}}\le\norm{(1+G_{F,\lambda})^{-1}},\]
		and the convergence of $G_{F_n,\lambda}$ to $G_{F,\lambda}$ in norm (\cref{prop:Bbound}),
		this yields
		\begin{align}
			\label{eq:conv.2}
			\Norm{\left(\HIBCl(F_n,V_n)+\ii\right)^{-1}-\left(\HIBCl(F,V_n)+\ii\right)^{-1}}
			\xrightarrow{n\to\infty}0.
		\end{align}
	Inserting \cref{eq:conv.1,eq:conv.2} into \cref{eq:resdecomp} proves the statement.
\end{proof}
Let us further analyze some properties of the domain of $\HIBCl(G,F)$.
We here restrict our attention to the two-nilpotent interactions relevant for our application, but emphasize that applications of Neumann series would yield similar statements for the more general case.

The following invariance statement in combination with the standard relative bound \cref{lem:simple.relbound} will allow us to extract the infrared part of the interaction as infinitesimal perturbation of $\HIBCl(G,F)$.
\begin{lem}\label{lem:IBCinv}
	Let $F\in\fb_2$ satisfy $F=F_>$ and $F(k)F(p)=0$ for $\mu$-almost all $k,p\in\cM$.
	Then for all $\lambda>0$ we have that $(1+G_{F,\lambda}^*)^{-1}\sD(\dG(\omega))\subset \sD(\dG(\omega_\le))$
	and
	\[ \|\dG(\omega_\le)\psi\| \le \big(1+\norm{G_{F,\lambda}}\big)^2\Big(1+\|T_{V,\lambda}(\dG(\omega)+\lambda)^{-1}\|\Big) \norm{\HIBCl(F,V)\psi}, \quad \psi\in\sD(\HIBCl(F,V)) \]
	for any $F$, $V$ and $\lambda$ satisfying the assumptions of \cref{prop:HIBC}.
\end{lem}
\begin{proof}
	We apply the commutation relation $[\dG(\omega_\le),G_{F,\lambda}^*]=0$ on $\HS_{\sff\sfi\sfn}$, which can easily be checked by direct calculation (and is a version of the canonical commutation relation $[\dG(A),\ad(f)] = \ad(Af)$). Since $\HS_{\sff\sfi\sfn}$ is a core for $\dG(\omega_\le)$,
	 this yields
	\begin{align*}
		\norm{\dG(\omega_\le)G_{F,\lambda}^*\phi} \le \norm{G_{F,\lambda}^*}\norm{\dG(\omega_\le)\phi} = \norm{G_{F,\lambda}}\norm{\dG(\omega_\le)\phi}, \quad \phi\in\sD(\dG(\omega_\le)).
	\end{align*}
	Thus, by \cref{lem:Gboundedinv}, the claimed domain inclusion holds and
	\begin{align}
		\norm{(\dG(\omega_\le)+\lambda)(1+G_{F,\lambda}^*)^{-1}\phi}
		&\le (1+\norm{G_{F,\lambda}})\norm{(\dG(\omega)+\lambda)\phi}
		\label{eq:relestimate}
		\\&\le (1+\norm{G_{F,\lambda}})\Big(1+\|T_{V,\lambda}(\dG(\omega)+\lambda)^{-1}\|\Big)\norm{(\dG(\omega) + \lambda - T_{V,\lambda})\phi}
		\nonumber.
	\end{align}
	Inserting $\phi = (1+G_{F,\lambda}^*)\psi$ then proves the statement.
\end{proof}
The following simple statement
will allow  us to compare the domain of $\HIBCl(F,V)$ to the (form) domain of $\dG(\omega)$, in turn allowing us to strengthen the convergence statement in \cref{mainthm}.

\begin{lem}\label{lem:Breg-one}
	If $F\in\fb_s$ for $s\in[1,2]$,
	then $\ran G_{F,\lambda}^*\subset \sD(\dG(\omega)^{1-\frac s2})$ for all $\lambda>0$. Further, for $r\le1-\frac s2$ and $\psi\in\sD(\dG(\omega)^{r})$, we have
	\begin{align*}
		\norm{(\dG(\omega)+\lambda)^{r}G_{F,\lambda}^*\psi} \le \norm{F}_{\fb_s}\lambda^{\frac {s}2-1}\norm{(\dG(\omega)+\lambda)^{r}\psi}.
	\end{align*}
\end{lem}
\begin{proof}
		Let $\alpha\le 1-\frac s2$ and $\beta\ge0$.
		Then for $\phi\in\sD(\dG(\omega)^\alpha)$ and $\lambda>0$, the estimate \cref{eq:a-sbound} implies
		\begin{align*}
			\|(\dG(\omega) + \lambda)^{-\beta} G_{F,\lambda}(\dG(\omega)+\lambda)^{\alpha}\phi\|\le \|F\|_{\fb_s}\lambda^{-\beta}\|(\dG(\omega)+\lambda)^{\frac{s}{2}+\alpha-1}\phi\|\,.
		\end{align*}
		Setting $\alpha=1-\frac s2$ and $\beta=0$, taking adjoints implies that $\dG(\omega)^{1-\frac s2}G_{F,\lambda}^*$ is bounded and thus the domain statement.
		
		Furthermore, for $\alpha=\beta=r$, we have
		\begin{align*}
			\norm{(\dG(\omega) + \lambda)^{-r} G_{F,\lambda}(\dG(\omega)+\lambda)^{r}} \le \norm{F}_{\fb_s}\lambda^{\frac s2-1}
		\end{align*}
		and thus the statement follows from the estimate.
		\begin{align*}
			\norm{(\dG(\omega)+\lambda)^{r}G_{F,\lambda}^*\psi}
			&\le \norm{(\dG(\omega) + \lambda)^r G_{F,\lambda}^*(\dG(\omega)+\lambda)^{-r}} \norm{(\dG(\omega)+\lambda)^r\psi}
			\\& = \norm{(\dG(\omega) + \lambda)^{-r} G_{F,\lambda}(\dG(\omega)+\lambda)^{r}} \norm{(\dG(\omega)+\lambda)^r\psi}
			.
			\qedhere
		\end{align*}
\end{proof}
Let us summarize the implications for $\HIBCl(F,V)$.
\begin{cor}\label{lem:Breg}
	Under the assumptions of \cref{prop:HIBC} and further assuming $F\in\fb_s$ for some $s\in[0,2]$ as well as $F(k)F(p)=0$ for $\mu$-almost all $k,p\in\cM$, we have $\sD(\HIBCl(F,V))\subset \sD(\dG(\omega)^{1-\frac s2})$ and 
	\begin{align*}
		\norm{\dG(\omega)^{1-\frac s2}\psi} \le \big(1+\norm{F}_{\fb_s}\lambda^{\frac s2-1}\big)^2 \Big(1+\|T_{V,\lambda}(\dG(\omega)+\lambda)^{-1}\|\Big)\norm{\HIBCl(F,V)\psi}
	\end{align*}
\end{cor}
	\begin{proof}
		The domain inclusion immediately follows from \cref{lem:Breg-one}.
		Further, by combining \cref{lem:Gboundedinv,lem:Breg-one}, we find
		\begin{align*}
			\norm{(\dG(\omega)+\lambda)^{1-\frac s2}(1+G_{F,\lambda}^*)^{-1}\phi}
			&\le \big(1+\norm{F}_{\fb_s}\lambda^{\frac s2-1}\big)\norm{(\dG(\omega)+\lambda)^{1-\frac s2}\phi}
			\\
			&\le \big(1+\norm{F}_{\fb_s}\lambda^{\frac s2-1}\big)\norm{(\dG(\omega)+\lambda)\phi}
		\end{align*}
		Noticing the similarity to \cref{eq:relestimate}, we can proceed as in the proof of \cref{lem:IBCinv}. Combined with the bound \cref{prop:Bbound}, this proves the statement.
	\end{proof}
\endgroup
\subsection{Spin Boson Models with 2-Nilpotent Interactions}\label{subsec:nilreg}
We now want to put the ultraviolet regular Hamiltonian defined in \cref{subsec:UVreg} into relation with the IBC Hamiltonian with 2-nilpotent ultraviolet part of the interaction.

The following normal-ordering statement is at the heart of the interior boundary condition method, also compare \cref{eq:Tsimple} in the introduction.
\begin{lem}\label{lem:Tnormord}
	Let $F\in\fb_0\cap\fb_1$. Then, for all $\lambda>0$,
	\begin{align*}
		T_{F,\lambda} = a(F)(\dG(\omega)+\lambda)^{-1}\ad(F) - \braket{F,F}_{\fb_1} \qquad\mbox{on}\ \sD(\dG(\omega)).
	\end{align*}
\end{lem}
\begin{proof}
	A direct calculation using the canonical commutation relations  (and the pull-through formula derived from these) yields
	\begin{align*}
		a(F)(\dG(\omega)+\lambda)^{-1}\ad(F)
		& = \int_{\cM^2} F(q)^*a_q(\dG(\omega)+\lambda)^{-1}\ad_pF(p)\d\mu(q)\d\mu(p)\\
		& = \int_{\cM^2} F(q)^*\ad_p(\dG(\omega)+\lambda +\omega(p)+\omega(q))^{-1}a_qF(p)\d\mu(q)\d\mu(p)\\
		& \qquad + \int_{\cM} F(q)^*(\dG(\omega)+\lambda +\omega(q))^{-1}F(q)\d\mu(q).
	\end{align*}
	The first line is easily identified as $\Theta_{1,\lambda}$, whereas for the second we obtain
	\begin{align*}
		&\int_{\cM} F(q)^*(\dG(\omega)+\lambda +\omega(q))^{-1}F(q)\d\mu(q) - \braket{F,F}_{\fb_1}
		\\&\qquad=
		\int_{\cM} F(q)^*(\dG(\omega)+\lambda +\omega(q))^{-1}-(\omega(q))^{-1}F(q)\d\mu(q)
		=
		\Theta_{0,\lambda}.\qedhere
	\end{align*}
\end{proof}
From there, we can rewrite the Hamiltonian with nilpotent interactions.
\begin{proof}[\textbf{Proof of \cref{thm:nilpotent}}]
	Using that $G_{V_\sfN,\lambda} = a(V_\sfN)(\dG(\omega)+\lambda)^{-1}$, we immediately obtain from the definition 
	\begin{align*}
		\HIBCl(V_\sfN)
		=\, &
		\dG(\omega) + \lambda + a(V_\sfN) + \ad(V_\sfN)
		\\&
		 + a(V_\sfN)(\dG(\omega)+\lambda)^{-1}\ad(V_\sfN) - T_{V_\sfN,\lambda}
		\\&
		 - a(V_\sfN)(\dG(\omega)+\lambda)^{-1}T_{V_\sfN,\lambda}(\dG(\omega)+\lambda)^{-1}\ad(V_\sfN).
	\end{align*}
	The first line on the right hand side of the above is easily recognized to be $\Hreg(0,V_\sfN) + \lambda$, cf. \cref{def:Hreg}.
	The second line is equal to $\braket{V_\sfN,V_\sfN}_{\fb_1}$, by \cref{lem:Tnormord}.
	Once more applying \cref{lem:Tnormord}, the third line equals
	\begin{align*}
		- a(V_\sfN)(\dG(\omega)+\lambda)^{-1}\big(a(V_\sfN)(\dG(\omega)+\lambda)^{-1}\ad(V_\sfN) - \braket{V_\sfN,V_\sfN}_{\fb_1}\big)(\dG(\omega)+\lambda)^{-1}\ad(V_\sfN)
	\end{align*}
	This expression
	 vanishes, since $V_\sfN(k)V_\sfN(p)$ vanishes almost everywhere and commutes with $\dG(\omega)$.
	This proves the statement.
\end{proof}
\section{Diagonalization of Normal Interactions}
\label{sec:normal}
In this \lcnamecref{sec:normal}, we treat the normal part of the interaction, by an appropriate dressing transformation.
In \cref{subsec:Weyl}, we will first collect some simple properties of the Weyl operators defined in \cref{def:Weyl}, which are standard in the case $\HS_\sfs=\IC$. We nevertheless give the simple proofs here for the convenience of the reader.
We then present the proof of \cref{mainthm} in \cref{subsec:dressren}, by combining results from \cref{sec:IBC,subsec:Weyl}.
In the final \cref{subsec:nonren}, we prove \cref{thm:nonren}.

\subsection{Transformation Behavior and Convergence of Weyl Operators}
\label{subsec:Weyl}
Let us discuss the transformation behavior of second quantization and field operators under our (generalized) Weyl operators.
\begin{lem}\label{lem:trans} Let $F,G \in \fb_0$.
	\begin{enumlem}
		\item\label{lem:trans.ph}
			If $[F,G]=[F,G^*] =[F,F^*]=[F,F]= 0$, then 
			\[W(F)\ph(G)W(F)^* = \ph(G) + \braket{F,G}_{\fb_0} + \braket{G,F}_{\fb_0}.\]
		\item\label{lem:trans.dG} If $F \in \fb_{-2}$, then on $\sD(\dG(\omega))\cap \sD(\ph(\omega F))$
			\[ W(F)\dG(\omega)W(F)^* = \dG(\omega) + \ph(\omega F) + \braket{F,F}_{\fb_{-1}}.\]
	\end{enumlem}
\end{lem}
\begin{proof}[{Proof of \subcref{lem:trans.ph}}]
	By the selfadjointness of both sides of the transformation formula, note that it suffices to prove the equality on any core for $\ph(G)$, in our case $\HS_{\sff\sfi\sfn}$.
	We remark that \cref{lem:simple.relbound} yields that all vectors in $\HS_{\sff\sfi\sfn}$ are analytic for $\ph(H)$, $H\in\fb_0$ and thus $W(H_1)^*\HS_{\sff\sfi\sfn}\subset \sD(\ph(H_2)\ph(H_3))$, $H_1,H_2,H_3\in\fb_0$. A straightforward approximation argument now also yields that the canonical commutation relation \cref{lem:simple.CCRph} hold on $W(H)^*\HS_{\sff\sfi\sfn}$.
	
	We now fix $\psi\in\HS_{\sff\sfi\sfn}$ and consider the function $f:[0,1]\to\HS$ given by $f(t)\coloneqq W(tF)\ph(G)W(tF)^*\psi$. Differentiating with respect to $t$ and applying the canonical commutation relation \cref{lem:simple.CCRph} yields
	\begin{align*}
		f'(t) &= \ii W(tF)\ph(\ii F)\ph(G)W(tF)^*\psi - \ii W(tF)\ph(G)\ph(\ii F)W(tF)^*\psi \\&= \ii W(tF)[\ph(\ii F),\ph(G)]W(tF)^*\psi \\& = \ii W(tF)(\braket{\ii F,G}_{\fb_0} - \braket{G,\ii F}_{\fb_0})W(tF)^*\psi \\&= W(tF)(\braket{F,G}_{\fb_0}+\braket{F,G}_{\fb_0})W(tF)^*\psi.
	\end{align*}
	Now since $\braket{F,G}_{\fb_0}+\braket{G,F}_{\fb_0}$ commutes with $F$ and thus with $\ph(F)$ and $W(tF)$, employing the unitarity of $W(tF)$ yields
	\[
	f'(t) = (\braket{F,G}_{\fb_0}+\braket{G,F}_{\fb_0})\psi.
	\]
	Combined with $f(0)=\ph(G)\psi$ and the fundamental theorem of calculus, this proves the statement.
\end{proof}
\begin{proof}[{Proof of \subcref{lem:trans.dG}}]
	We proceed similarly as in the previous proof and fix the function $g:[0,1]\to\HS$ with $g(t)\coloneqq W(tF)\dG(\omega)W(tF)^*\psi$ for given $\psi\in\HS_{\sff\sfi\sfn}$.
	Differentiating and employing the canonical commutation relation \cref{lem:simple.CCRdG} now yields
	\[ g'(t) = \ii W(tF)[\ph(\ii F),\dG(\omega)]W(tF)^*\psi = W(tF)\ph(\omega F)W(tF)^*\psi.  \]
	Inserting \cref{lem:trans.ph}, we arrive at
		\[ g'(t) = (\ph(\omega F) + 2t\braket{F,\omega F}_{\fb_0})\psi = (\ph(\omega F) + 2t\braket{F,F}_{\fb_{-1}})\psi. \]
	Combined with $g(0)=\dG(\omega)\psi$ and the fundamental theorem of calculus, this yields the statement on $\HS_{\sff\sfi\sfn}$.
	Since the left hand side of the transformation formula is selfadjoint and the right hand side is $\dG(\omega)$-bounded, the statement follows.
\end{proof}
We will also require the following convergence results for Weyl operators. The statement is a generalization, e.g., of \cite{GriesemerWuensch.2018} or \cite{MatteMoller.2018}.
\begin{lem}\label{lem:Weylconv}
	For $F,G\in \fb_0$ such that $[F,G]=[F,G^*]=[F,F]=[F,F^*]=0$ and $\psi\in\sD(\ph(F)\cap \ph(G))$, we have
		\[ \|(W(F)-W(G))\psi\| \le \|\ph(F-G)\psi\| + \frac{1}{2}\|(\braket{F,F-G}_{\fb_0}+\braket{F-G,F}_{\fb_0})\psi\|. \]
	Thus $\fb_0\ni F\mapsto W(F)\in\cB(\HS)$ is continuous w.r.t. the strong operator topology in $\cB(\HS)$.
	Further, if $F,G\in \fb_0\cap\fb_1$, then for all $\theta\in[0,1]$
		\begin{align*} &\|(W(F)-W(G))(1+\dG(\omega))^{-\theta/2}\| \\&\qquad\qquad\le 2^{1-\theta}\big(4(\|F-G\|_{\fb_0}\vee\|F-G\|_{\fb_1}) + \frac{1}{2}\|(\braket{F,F-G}_{\fb_0}+\braket{F-G,F}_{\fb_0})\|_{\cB(\HS_\sfs)}\big)^{\theta}.\end{align*}
\end{lem}
\begin{proof}
	First notice that it suffices to prove the statement for $\theta=1$, by the bound $\norm{W(F)-W(G)}\le 2$ and the interpolation property of the two Hilbert scales given by $\HS$ equipped with the standard norm and the  norm given by $\norm{(1+\dG(\omega))^{-\theta/2}\bullet }$ with respect to each other, see for example \cite[Theorem~9.1]{KreinPetunin.1966}.
	
		Now since $W(G)$ is unitary, we can write
		\begin{align*}
			\|(W(F)-W(G))\psi\|=\|W(G)^*W(F)\psi-\psi\|\,.
		\end{align*}
		Now, we define $h(t)=W(tG)^*W(tF)\psi$ and by differentiating obtain
		\begin{align*}
			h'(t)=\ii W(tG)^*\ph(F-G)W(tF)\psi.
		\end{align*}
		Again by the unitarity of both $W(F)$ and $W(G)$ as well as \cref{lem:trans.ph}, for which the commutation assumptions are satisfied, this implies
		\begin{align*}
			\|h'(t)\|=\|W(tF)^*\varphi(F-G)W(tF)\psi\|\leq\|\varphi(F-G)\psi\|+t\|(\langle F,F-G\rangle_{\fb_0}
			+\langle F-G,F\rangle_{\fb_0})\psi\|.
		\end{align*}
		The fundamental theorem of calculus gives the desired estimate. Finally, the strong continuity statement follows directly from the density of $\sD(\varphi(F-G))$ and the last norm bound is a consequence of \cref{lem:simple.relbound}.
	\end{proof}

\subsection{Renormalization with Dressing Transformations}
\label{subsec:dressren}
We can now move to the proof of \cref{mainthm}, applying our dressing transformation.
Thus, let us in this section work under the assumptions and notation of that \lcnamecref{mainthm}, i.e., we especially assume $S$, $V$ and $s_\sfN\in[1,2]$ to  be as defined therein.

We split the proof into several lemmas, starting with the definition of $H(S,V)$.
Since the $\lambda$-independence of our definition is a priori
not clear, we will here denote it by $\tilde H_\lambda(S,V)$
\begin{lem}
	\label{lem:HtildeSA}
	The operator
	\begin{align*}
		\tilde H_\lambda(S,V) \coloneqq S + W(\omega^{-1}V_\sfD)(\HIBCl(V_\sfN) + \ph(V_\le))W(\omega^{-1}V_\sfD)^* - \lambda.
	\end{align*}
	is selfadjoint and lower-semibounded for all $\lambda>0$.
\end{lem}
\begin{proof}
	First, recalling that either $s_\sfN<2$ or $\norm{V_\sfN}_{\fb_2}<\frac12$, we note that $\HIBCl(V_\sfN)$ is selfadjoint and lower-semibounded by \cref{prop:HIBC}. Furthermore, by \cref{lem:simple.relbound,lem:IBCinv}, $\ph(V_\le)$ is infinitesimally $\HIBCl(V_\sfN)$-bounded and hence the operator $\HIBCl(V_\sfN)+\ph(V_\le)$ is selfadjoint and lower-semibounded by the Kato--Rellich theorem.
	Thus, the statement follows from the unitarity of $W(\omega^{-1}V_\sfD)$.
\end{proof}
We now relate $\tilde H_\lambda(S,V)$ to $\Hreg(S,V)$.
\begin{lem}\label{lem:Htildereg}
	 If $V_\sfN,V_\sfD\in\fb_0$, then for all $\lambda>0$
	\begin{align*}
		\tilde H_\lambda(S,V) = \Hreg(S,V) + \braket{V_>,V_>}_{\fb_1}.
	\end{align*}
\end{lem}
\begin{proof}
	By 	combining \cref{thm:nilpotent,lem:trans} and recalling the definition \cref{eq:bs} as well as the assumption $V_>=V_\sfD+V_\sfN$, we find
	\begin{align*}
		\tilde H_\lambda(S,V) &= S+ W(\omega^{-1}V_\sfD)(\Hreg(0,V_\sfN) + \ph(V_\le) + \lambda + {\braket{V_\sfN,V_\sfN}_{\fb_1}})W(\omega^{-1}V_\sfD)^* - \lambda\\
		&= S + W(\omega^{-1}V_\sfD)\Hreg(0,V_\le+V_\sfN)W(\omega^{-1}V_\sfD)^* +\braket{V_\sfN,V_\sfN}_{\fb_1}\\
		& = S + \Hreg(0,V) + \braket{\omega^{-1}V_\sfD,\omega^{-1}V_\sfD}_{\fb_{-1}} + \braket{V_{\sfN},\omega^{-1}V_\sfD}_{\fb_0} + \braket{\omega^{-1}V_\sfD,V_\sfN}_{\fb_0} +\braket{V_\sfN,V_\sfN}_{\fb_1}\\
		& = \Hreg(S,V) + \braket{V_\sfD,V_\sfD}_{\fb_1} + \braket{V_\sfN,V_\sfD}_{\fb_1} + \braket{V_\sfD,V_\sfN}_{\fb_1} + \braket{V_\sfN,V_\sfN}_{\fb_1}\\
		& = \Hreg(S,V) + \braket{V_>,V_>}_{\fb_1}.
	\end{align*}
	Here, we used $[V_\sfD,V_\sfN]=0$ and $[V_\sfD^*,V_\sfN]=0$ in the second equality, where the latter follows from the fact that $V_\sfD$ is normal.
\end{proof}
Let us also study the continuity of $\tilde H_\lambda(S,V)$ with respect to its arguments.
\begin{prop}\label{lem:Htildeconv}
	Let $(F_n=F_{n,\le} + F_{n,\sfD} + F_{n,\sfN})_{n\in\IN}$ be a sequence satisfying the assumptions of \cref{mainthm} with $V=F_n$ for all $n\in\IN$ such that $F_{n,\le}\xrightarrow{\fb_1}V_\le$ and $F_{n,\#}\xrightarrow{\fb_2}V_\#$ for $\#\in\{\sfD,\sfN\}$.
	Further let $(S_n)_{n\in\IN}\subset \cB(\HS)$ be a strongly convergent sequence with limit $S$.
	 Then $\tilde H_\lambda(S_n,F_n)$ converges to $\tilde H_\lambda(S,V)$ in the strong resolvent sense for all $\lambda>0$. The convergence is in norm resolvent sense, if $S_n$ converges in norm and either $F_{n,\sfD} = 0$ for all $n\in\IN$ or $F_{n,\sfN}\xrightarrow{\fb_{s}}V_\sfN\in\fb_s$ for some $s<2$. 
\end{prop}
\begin{proof}We split the proof into three steps.
	
	\smallskip
	\noindent{\em Step 1.}
	We prove that $\HIBCl(F_{n,\sfN}) + \ph(F_{n,\le})$ converges to $\HIBCl(V_\sfN) + \ph(V_\le)$ in the norm resolvent sense.
	
	Note that selfadjointness of these operators is again clear, in view of \cref{lem:simple.relbound,lem:IBCinv}.
	Furthermore, these bounds imply that we can pick $\alpha>0$ sufficiently large that
	\begin{align*}
		&\sup_{n\in\IN}\norm{\ph(F_{n,\le})(\HIBCl(F_{n,\sfN})+\alpha \ii)^{-1}}<1,\\
		&\ph(F_{n,\le})(\HIBCl(F_{n,\sfN})+\alpha \ii)^{-1} \xrightarrow{n\to\infty} \ph(V_{\le})(\HIBCl(V_{\sfN})+\alpha \ii)^{-1} \quad \text{in norm.}
	\end{align*}
	Thus, the claim follows from \cref{prop:HIBC} and the identity
	\begin{align*}
		\big(\HIBCl(F_{n,\sfN})+\ph(F_{n,\le})+\alpha \ii\big)^{-1}
		=
		(\HIBCl(F_{n,\sfN})+\alpha \ii)^{-1}\big(1+\ph(F_{n,\le})(\HIBCl(F_{n,\sfN})+\alpha \ii)^{-1}\big)^{-1},
	\end{align*}
	which is a consequence of the Neumann series and the usual resolvent identity.
	
	\smallskip
	\noindent{\em Step 2.}
	We prove that $\tilde H_\lambda(0,F_{n})$ converges to $\tilde H_\lambda(0,V)$ in the sense mentioned in the statement.
	
	The strong resolvent convergence immediately follows from the strong continuity of the Weyl operators (\cref{lem:Weylconv}) combined with the norm continuity statement from Step 1 and the upper bound \cref{prop:Bbound}.
	Similarly, norm resolvent convergence in the case $F_{n,\sfD}=0$ is immediate. If $F_{n,\sfN}\in\fb_s$ for some $s<2$, by \cref{def:HIBCnew,lem:Breg}, we have $\sD(\HIBCl(F_{n,\sfN}))\subset \sD(\dG(\omega)^{1-s/2})$ and thus
	\begin{align*}
		&(\tilde H_\lambda(0,F_n)  +\ii)^{-1} - (\tilde H_\lambda(0,V)+\ii)^{-1}
		\\
		&= W(\omega^{-1}F_{n,\sfD})^*\left((\HIBCl(F_{n,\sfN})+\ph(F_{n,\le})+\ii)^{-1}-(\HIBCl(V_{\sfN}) + \ph(V_{\le})+\ii)^{-1}\right)W(\omega^{-1}V_\sfD)
		\\
		&\qquad + W(\omega^{-1}F_{n,\sfD})^*(\HIBCl(F_{n,\sfN})+\ph(F_{n,\le})+\ii)^{-1}(W(\omega^{-1}F_{n,\sfD})-W(\omega^{-1}V_\sfD))
		\\
		&\qquad + (W(\omega^{-1}F_{n,\sfD})-W(\omega^{-1}V_\sfD))^*(\HIBCl(V_\sfN) + \ph(V_{\le})+\ii)^{-1}W(\omega^{-1}V_\sfD).
	\end{align*}
	The first line converges to zero in norm by Step 1. To prove convergence of the second and third line to zero in norm,
	we note that $(\dG(\omega)+1)^{1-s/2}(\HIBCl(F_{n,\sfN}) + \ph(F_{n,\le})+\ii)^{-1}$ is uniformly bounded, by \cref{lem:Breg}. Thus, combining with \cref{lem:Weylconv} finishes this step.
	
	\smallskip
	\noindent{\em Step 3.}
	It remains to prove that $\tilde H_\lambda(S_n,F_{n})$ converges to $\tilde H_\lambda(S,V)$ in the sense mentioned in the statement. This follows similar to Step  1 from the identity
	\[
		(\tilde H_\lambda(S_n,F_n) + \alpha \ii)^{-1} = (\tilde H_\lambda(0,F_n)+\alpha \ii)^{-1}\big(1+S_n(\tilde H_\lambda(0,F_n)+\alpha \ii)^{-1}\big)^{-1}.
		\qedhere
	\]
\end{proof}
We can now conclude the
\begin{proof}[\textbf{Proof of \cref{mainthm}}]
	We proved the selfadjointness and lower-semiboundedness in \cref{lem:HtildeSA}. The independence of $\lambda$ as well as the convergence statement follow from \cref{lem:Htildereg,lem:Htildeconv}, using that $V_n = \chr_{\{\omega\le n\}}V\in\fb_0\cap \fb_1$ is always a suitable approximating sequence of form factors.
\end{proof}
\subsection{Proof of Non-Renormalizability}\label{subsec:nonren}
In this \lcnamecref{subsec:nonren}, we prove \cref{thm:nonren}.

Let us start by investigating non-renormalizability for the van Hove model \cite{vanHove.1952}.
The following lemma is well-known. We use notation from \cref{rem:bsinclusion}.
\begin{lem}[{\cite[Lemma 5.6]{DamMoller.2018b}}]
	\label{lem:ResConvNonRen}
	Let $\cH_s=\BC$ and $(v_n)_{n\in\IN}\subset\fb_0^>$ be a sequence of measurable functions
	such that $\|\omega^{-1}v_n\|_2=\norm{v_n}_{\fb_2}\xrightarrow{n\to\infty}\infty$. Then
	\begin{align*}
		W(v_n\omega^{-1})^*\left(\varepsilon(-1)^{\dG(1)}+\dG(\omega)+\varphi(v_n)+\|v_n\|^2_{\fb_1}\right)W(v_n\omega^{-1})\xrightarrow{n\to\infty} \dG(\omega)
	\end{align*}
	in the norm resolvent sense for every $\varepsilon\in\IR$.
\end{lem}
\begin{proof}[Strategy of proof]
	The weak convergence of $W(v_n\omega^{-1})^*(-1)^{\dG(1)}W(v_n\omega^{-1})$ to zero is easy to check, by a trial state argument.
	From there, the {\em strong} resolvent convergence can also directly be deduced. The hard part of the proof in \cite{DamMoller.2018b} is the extension of the statement to norm resolvent convergence and requires a momentum discretization argument.
\end{proof}
The norm resolvent convergence is essential to the next statement.
It is an adaption of \cite[Proof of Corollary 4.4]{DamMoller.2018a} to the van Hove model.
\begin{prop}\label{prop:vanHove}
	Let $\HS_\sfs=\IC$ and assume $v=v_>:\cM\to\IC$ measurable satisfies $v_>\notin\fb_2$. Then $\dG(\omega) + \ph(v_n) + E_n$ does not converge to a selfadjoint operator in strong resolvent sense for any sequence $(v_n)\subset \fb_0^>$ with $v_n\to v$ pointwise and any sequence $(E_n)\subset \IR$ such that $\liminf E_n-\|v_n\|_{\fb_1} > -\infty$.
\end{prop}
\begin{proof}
	Let us first observe that we can assume $E_n = \norm{v_n}_{\fb_1}$.
	This is obvious if there exists a subsequence such that $E_{n_k}-\norm{v_{n_k}}_{\fb_1}$ converges as $k\to\infty$, because shifts by convergent multiples of the identity preserve strong resolvent convergence.
	However, if this is not the case, then $\liminf E_n-\norm{v_n}_{\fb_1} = \infty$ and hence $\sigma(\dG(\omega) + \ph(v_n) + E_n)\xrightarrow{n\to\infty}\emptyset$, which implies that the spectrum of the limiting operator is empty \cite[Thm.~VIII.24]{ReedSimon.1972}, a contradiction to the selfadjointness of the limit.
	
	Now assume that $H$ is the strong resolvent limit of $\dG(\omega) + \ph(v_n) + \norm{v_n}_{\fb_1}$.
	Then, using \cref{lem:trans,lem:ResConvNonRen} and Fatou's lemma which implies $\norm{v_n}_{\fb_2}\xrightarrow{n\to\infty}\infty$, we find
		\begin{align*}
			&\Norm{(\dG(\omega) + \ph(v_n) + \norm{v_n}_{\fb_1} + \ii)^{-1} - (\eps(-1)^{\dG(1)} + \dG(\omega) + \ph(v_n) + \norm{v_n}_{\fb_1} + \ii)^{-1}}
			\\
			&\qquad = 
			\Big\|\left(W(\omega^{-1}v_n)^*(\dG(\omega)+\ph(v_n))W(\omega^{-1}v_n)+\norm{v_n}_{\fb_1}+\ii\right)^{-1}\\ 
			& \qquad\qquad  - \left(W(\omega^{-1}v_n)^*(\eps(-1)^{\dG(1)} + \dG(\omega)+\ph(v_n))W(\omega^{-1}v_n)+\norm{v_n}_{\fb_1}+\ii\right)^{-1}\Big\|\\
			&
			\qquad=\Big\|\left(\dG(\omega)+\ii\right)^{-1}  - \left(W(\omega^{-1}v_n)^*(\eps(-1)^{\dG(1)} + \dG(\omega)+\ph(v_n)+\norm{v_n}_{\fb_1})W(\omega^{-1}v_n)+\ii\right)^{-1}\Big\|
			\\
			&\qquad \xrightarrow{n\to\infty} 0,
		\end{align*}
	so
	$H$ is also the strong resolvent limit of $\eps(-1)^{\dG(1)} + \dG(\omega) + \ph(v_n) +\norm{v_n}_{\fb_1}$ for any $\eps\in\IR$.
	
	Thus, using the resolvent identity
	\begin{align*}
		(H+\ii)^{-1}&\Gamma(-1)(H+\ii)^{-1}
		\\
		& = \slim_{n\to\infty}\Big(
		(\dG(\omega) + \ph(v_n) + E_n + \ii)^{-1}
		-
		((-1)^{\dG(1)} + \dG(\omega) + \ph(v_n) + E_n + \ii)^{-1}\Big)
		\\
		&= 0,
	\end{align*}
	which is the desired contradiction, since the first line is injective as composition of injective operators.
\end{proof}
We can now give the
\begin{proof}[\textbf{Proof of \cref{thm:nonren}}]
	Assume towards contradiction that under the given assumptions $\Hreg(S,V_\leq+V_{\sfD,n})+B_n$ converges in the strong resolvent sense.
	
	We first observe that without loss of generality we can set $S=V_\leq=0$.
	This follows similarly to the Steps 1 and 3 in the proof of \cref{mainthm}.
	Hence, let $H_n\coloneqq \Hreg(0,V_{\sfD,n}) + B_n$ and assume that $H$ is the strong resolvent limit of $H_n$.	
	
	Now define the orthogonal projection $P_\psi = \ket\psi\bra{\psi}\otimes \Id_\FS$ on the subspace $\{\psi\otimes \phi|\phi\in\FS\}$, where $\psi$ is defined as in the statement of the theorem, and note that $\ran P_\psi\cong \FS(L^2(\cM,\mu))$, cf. \cref{def:Fock}. We easily see from the assumptions that
	\begin{align*}
		P_\psi H_n = \dG(\omega) + \ph(v_n) + \braket{\psi,B_n\psi} =  H_nP_\psi.
	\end{align*}
	Thus, by \cref{prop:AppenProjectStrongResolven}, $P_\psi H_n\upharpoonright_{\ran P_\psi}$ converges to $P_\psi H\upharpoonright_{\ran P_\psi}$ in the strong resolvent sense.
	
	This now stands in immediate contradiction to \cref{prop:vanHove}.
	To verify the assumptions therein, we only need to observe that
	\begin{align}\label{eq:Bnexp}
		\liminf_{n\to\infty}\big(\!\braket{\psi,B_n\psi}-\norm{v_n}_{\fb_1}^2)>-\infty.
	\end{align}
	The above now follows from our assumed lower bound on $\braket{\psi,B_n\psi}$ combined with the fact that
	\[\inf \sigma(\dG(\omega) + \ph(v)) = -\norm{v}_{\fb_1} \qquad\text{in the case $\HS_\sfs = \IC$, $v\in\fb_1$},\]
	  see for example \cite[\S13]{Arai.2018}, since by inserting trial states of the form $\psi\otimes \phi$ with $\phi\in\sD(\dG(\omega))$, $\norm\phi=1$ they imply the upper bound
	\begin{align*}
		\inf\sigma(\Hreg(S,V_\le + V_{\sfD,n})) \le
		\braket{\psi,S\psi} - \norm{\braket{\psi,V_\le \psi}+v_n}_{\fb_1}
		\le
		\braket{\psi,S\psi} + \norm{\braket{\psi,V_\le \psi}}_{\fb_1} - \norm{v_n}_{\fb_1}
		.
	\end{align*}
	This proves \cref{eq:Bnexp}, thus the desired contradiction and hence
	the operator $H$ can not exist.
\end{proof}

\appendix
\section{Operator Identities for Transformed Operators}
In this appendix, we collect some useful identities for non-unitary transformations of operators, which we could not find in the literature in exactly the same way.

The first statement is the transformation property for adjoints, which we apply in our selfadjointness proof for the IBC operators.
\begin{lem}[Adjoint of Transformed Operators]\label{lem:adjointtrans}
	Let $\cK$ be a Hilbert space and let $C$ be a closed, densely defined operator on $\cK$. Further assume that $A\in\cB(\cK)$ has bounded inverse.
	Then
	\begin{align*}
		(ACA^*)^* = AC^*A^*.
	\end{align*}
\end{lem}
\begin{proof}
	In view of \cite[Prop.~1.7]{Schmudgen.2012}, we already know
	\begin{align*}
		(ACA^*)^* =(CA^*)^*A^* \supset AC^*A^*.
	\end{align*}
	It thus remains to prove $\sD((CA^*)^*)\subset \sD(AC^*)=\sD(C^*)$.
	To this end let $x\in \sD((CA^*)^*)$ and $y\in \sD(C)$. Then
	\begin{align*}
		\braket{x,Cy} = \braket{x,CA^*(A^*)^{-1}y}=\braket{(CA^*)^*x,(A^*)^{-1}y}.
	\end{align*}
	Hence, $y\mapsto \braket{x,Cy}$ is continuous, by the Cauchy--Schwarz inequality, and thus the statement follows.
\end{proof}
We also apply the following version of the resolvent identity for non-unitary transformations of selfadjoint operators.
\begin{lem}[Resolvent Identity for Transformed Operators]\label{lem:residentity}
	Let $\cK$ be a Hilbert space and let $T$ be selfadjoint on $\cK$. Further assume that $A,B\in\cB(\cK)$ have bounded inverse. Then
	\begin{align*}
		&
		(ATA^*+\ii)^{-1} - (BTB^*+\ii)^{-1}
		\\&
		=
		(ATA^*+\ii)^{-1}(B-A)TB^*(BTB^*+\ii)^{-1} + \left(TA^*(ATA^*+\ii)^{-1}\right)^*(B^*-A^*)(BTB^*+\ii)^{-1}
	\end{align*}
\end{lem}
\begin{proof}
	Let $(T_n)_{n\in\IN}\in\cB(\cK)$ satisfy $T_n\psi\to T\psi$, $\psi\in\sD(T)$ (existence can, e.g., be deduced from the spectral theorem).
	Then the claim with $T=T_n$ follows from the usual resolvent identity. 
	Further, our convergence assumption implies $(AT_nA^*+\ii)^{-1}\to (ATA^*+\ii)^{-1}$ in the strong operator topology, cf. \cite[Theorem VIII.25]{ReedSimon.1972}.
	Together with the uniform bound $\norm{T_nA^*(AT_nA^*+\ii)^{-1}}\le \norm{A^{-1}}$ and the same observations with $A$ replaced by $B$, both sides of the identity converge as $n\to\infty$.
\end{proof}

\section{Strong Resolvent Convergence and Reducing Subspaces}
\label{app:redsubspace}
In this \lcnamecref{app:redsubspace}, we study the invariance of reducing subspaces of selfadjoint operators under strong resolvent convergence. The material seems to be standard, but is not presented in any textbook in this form, to the authors' knowledge.

Let us first recall the meaning of projections for reducing subspaces.
\begin{lem}
	Let $A$ be a selfadjoint on a Hilbert space $\cK$ and $U\subset\cK$ a (non trivial) closed subspace such that the orthogonal projection $P_U$ on $U$ satisfies $P_UA=AP_U$. Then the restriction $\tilde{A}:=A\!\upharpoonright_U:\sD(A)\cap U\to U$ is selfadjoint (on $U$).
\end{lem}
\begin{proof} Using that $P_UA=AP_U$ it is easy to see that $\sD(\tilde{A})=\sD(A)\cap U=P_U\sD(A)$ is dense in $U$ and that $\tilde{A}$ is a closed symmetric operator. Therefore, to prove selfadjointness it is enough to see that $\tilde{A}\pm i$ are surjective. Take $\psi\in U$, since $A$ is selfadjoint, the operators $A\pm i $ are surjective and we can find $\varphi_\pm\in\sD(A)$ such that 
	\begin{align*}
			(A\pm i)\varphi_{\pm}=\psi\,.
		\end{align*}
		This implies that
		\begin{align*}
			(\tilde{A}\pm i)P_U\varphi_\pm=(A\pm i)P_U\varphi_\pm=P_U(A\pm i)\varphi_\pm=P_U\psi=\psi\,,
		\end{align*}
		which gives that $\tilde{A}\pm i:\sD(\tilde{A})\longrightarrow U$ is surjective. 
\end{proof}
We now prove that a reducing subspace of a sequence remains reducing for the strong resolvent limit.
\begin{lem}
	Let $A_n$ be a sequence of selfadjoint operators on a Hilbert space $\cK$ converging to a selfadjoint operator $A$ in the strong resolvent sense.
	We further assume that there exists $U\subset\cK$ a (non trivial) closed subspace such that $P_UA_n=A_nP_U$.
	Then $AP_U=P_UA$.
\end{lem}
\begin{proof}
	Let $\psi\in\sD(A)$. Then, since strong resolvent convergence is equivalent to strong graph convergence, cf. \cite[Thm.~VIII.26]{ReedSimon.1972}, there exists $\psi_n\in\sD(A_n)$ such that
	\begin{align*}
		\psi_n\longrightarrow \psi\quad\text{and}\quad A_n\psi_n\longrightarrow A\psi\,.
	\end{align*}
	We now notice that $A_nP_U=P_UA_n$ in particular implies that $P_U\sD(A_n)\subset\sD(A_n)$. Therefore, $P_U\psi_n\in\sD(A_n)$, and since $P_U$ is continuous, we have that 
	\begin{align*}
		P_U\psi_n\longrightarrow P_U\psi\,.
	\end{align*}
	We now prove that the sequence $A_nP_U\psi_n$ is convergent. Indeed, by writing
	\begin{align*}
		\|A_n\psi_n-A_m\psi_m\|^2=\|P_U(A_n\psi_m-A_m\psi_m)\|^2+\|P_{U^\perp}(A_n\psi_m-A_m\psi_m)\|^2\,,
	\end{align*}
	the assumption implies that
	\begin{align*}
		\|A_nP_U\psi_n-A_mP_U\psi_m\|=\|P_U(A_n\psi_n-A_m\psi_m)\|\leq\|A_n\psi_n-A_m\psi_m\|\,,
	\end{align*}
	so the convergence of $(A_n\psi_n)$ implies that $(A_nP_U\psi_n)$ is Cauchy and thus convergent. We have obtained that
	\begin{align}
		P_U\psi_n\longrightarrow P_U\psi\,,\quad\text{with }\,P_U\psi_n\in\sD(A_n)\quad\text{and}\quad A_n(P_U\psi_n)\longrightarrow \xi\,.
	\end{align}
	This implies, again by the strong graph convergence, that $P_U\psi\in\sD(A)$ and $\xi=AP_U\psi$\,. Finally, the continuity of $P_U$ gives that
	\begin{align*}
		&P_UA\psi=\lim_nP_UA_n\psi_n=\lim_n A_nP_U\psi_n= AP_U\psi\,. \qedhere
	\end{align*}
\end{proof}
It remains to observe that strong resolvent convergence carries over to the reduced operators.
\begin{prop}\label{prop:AppenProjectStrongResolven}
	Let $A_n$ be a sequence of self-adjoint operators converging in strong resolvent sense to a self-adjoint operator $A$ on a Hilbert space $\cK$. We further assume that there exists $U\subset\cK$ a (non trivial) closed subspace such that $P_UA_n=A_nP_U$. Then the selfadjoint operators $A_n\!\upharpoonright_U$ converge to the selfadjoint operator $A\!\upharpoonright_U$ in the strong resolvent sense.
	\begin{proof}
		We first notice that by the selfadjointness is obvious, by the previous lemmas. Therefore, proving strong resolvent convergence is the same as proving strong graph convergence.
		
		Let $\tilde{\psi}\in \sD(A)\cap U$, then there exists $\psi\in \sD(A)$ such that $\tilde{\psi}=P_U\psi$. We also know that there exists a sequence $\psi_n\in\sD(A_n)$ such that $A_n\psi_n\longrightarrow A\psi$. Therefore we have that
		\begin{align*}
			\tilde{\psi_n}:=P_U\psi_n\longrightarrow P_U\psi=\tilde{\psi}\qquad\text{and}\qquad A_n\tilde\psi_n=A_nP_U\psi_n=P_UA_n\psi_n\longrightarrow P_UA\psi=AP_U\psi=\tilde{A}\tilde{\psi}\,.
		\end{align*}
		This implies that the operator $A\!\upharpoonright_U$ is the strong graph limit of $A_n\!\upharpoonright_U$
		and thus the statement.
	\end{proof}
\end{prop}


\bibliographystyle{halpha-abbrv}
\bibliography{00lit}

\newcommand{\etalchar}[1]{$^{#1}$}
\begin{thebibliography}{BBKM17}
\expandafter\ifx\csname url\endcsname\relax
  \def\url#1{\texttt{#1}}\fi
\expandafter\ifx\csname doi\endcsname\relax
  \def\doi#1{\burlalt{doi:#1}{http://dx.doi.org/#1}}\fi
\expandafter\ifx\csname urlprefix\endcsname\relax\def\urlprefix{URL }\fi
\expandafter\ifx\csname href\endcsname\relax
  \def\href#1#2{#2}\fi
\expandafter\ifx\csname burlalt\endcsname\relax
  \def\burlalt#1#2{\href{#2}{#1}}\fi

\bibitem[AH95]{AraiHirokawa.1995}
A.~Arai and M.~Hirokawa.
\newblock On the Existence and Uniqueness of Ground States of the Spin-Boson
  {H}amiltonian.
\newblock {\em Hokkaido Univ. Prepr. Ser. Math.}, 309:2--20, 1995.
\newblock \doi{10.14943/83456}.

\bibitem[AH97]{AraiHirokawa.1997}
A.~Arai and M.~Hirokawa.
\newblock On the Existence and Uniqueness of Ground States of a Generalized
  Spin-Boson Model.
\newblock {\em J. Funct. Anal.}, 151(2):455--503, 1997.
\newblock \doi{10.1006/jfan.1997.3140}.

\bibitem[AM21]{AlvarezMoller.2022}
B.~Alvarez and J.~S. M{\o}ller.
\newblock Ultraviolet Renormalisation of a quantum field toy model I.
\newblock Preprint, 2021,
  \burlalt{arXiv:2103.13770}{http://arxiv.org/abs/2103.13770}.

\bibitem[AM23]{AlvarezMoller.2023}
B.~Alvarez and J.~S. M{\o}ller.
\newblock Ultraviolet Renormalisation of a Quantum Field Toy Model II.
\newblock Preprint, 2023,
  \burlalt{arXiv:2312.10496}{http://arxiv.org/abs/2312.10496}.

\bibitem[Ama91]{Amann.1991}
A.~Amann.
\newblock Ground states of a spin-boson model.
\newblock {\em Ann. Phys.}, 208(2):414--448, 1991.
\newblock \doi{10.1016/0003-4916(91)90302-O}.

\bibitem[Amm00]{Ammari.2000}
Z.~Ammari.
\newblock Asymptotic Completeness for a Renormalized Nonrelativistic
  {H}amiltonian in Quantum Field Theory: The {N}elson Model.
\newblock {\em Math. Phys. Anal. Geom.}, 3(3):217--285, 2000.
\newblock \doi{10.1023/A:1011408618527}.

\bibitem[Ara18]{Arai.2018}
A.~Arai.
\newblock {\em Analysis on {F}ock Spaces and Mathematical Theory of Quantum
  Fields}.
\newblock World Scientific, New Jersey, 2018.
\newblock \doi{10.1142/10367}.

\bibitem[BBKM17]{BachBallesterosKoenenbergMenrath.2017}
V.~Bach, M.~Ballesteros, M.~K{\"o}nenberg, and L.~Menrath.
\newblock Existence of ground state eigenvalues for the spin–boson model with
  critical infrared divergence and multiscale analysis.
\newblock {\em J. Math. Anal. Appl.}, 453(2):773--797, 2017,
  \burlalt{arXiv:1605.08348}{http://arxiv.org/abs/1605.08348}.
\newblock \doi{10.1016/j.jmaa.2017.03.075}.

\bibitem[BDH19]{BallesterosDeckertHaenle.2019b}
M.~Ballesteros, D.-A. Deckert, and F.~H\"{a}nle.
\newblock Analyticity of resonances and eigenvalues and spectral properties of
  the massless Spin-Boson model.
\newblock {\em J. Funct. Anal.}, 276(8):2524--2581, 2019,
  \burlalt{arXiv:1801.04021}{http://arxiv.org/abs/1801.04021}.
\newblock \doi{10.1016/j.jfa.2019.02.008}.

\bibitem[BDP12]{BachmannDeckertPizzo.2012}
S.~Bachmann, D.-A. Deckert, and A.~Pizzo.
\newblock The mass shell of the {N}elson model without cut-offs.
\newblock {\em J. Funct. Anal.}, 263(5):1224--1282, 2012,
  \burlalt{arXiv:1104.3271}{http://arxiv.org/abs/1104.3271}.
\newblock \doi{10.1016/j.jfa.2012.04.021}.

\bibitem[BH22]{BachHach.2022}
V.~Bach and A.~Hach.
\newblock On the Ultraviolet Limit of the {P}auli-{F}ierz {H}amiltonian in the
  {L}ieb-{L}oss Model.
\newblock {\em Ann. Henri Poincar\'{e}}, 23(6):2207--2245, 2022,
  \burlalt{arXiv:2004.06494}{http://arxiv.org/abs/2004.06494}.
\newblock \doi{10.1007/s00023-021-01124-2}.

\bibitem[BHKP25]{BetzHinrichsKraftPolzer.2025}
V.~Betz, B.~Hinrichs, M.~N. Kraft, and S.~Polzer.
\newblock On the {I}sing Phase Transition in the Infrared-Divergent Spin Boson
  Model.
\newblock Preprint, 2025,
  \burlalt{arXiv:2501.19362}{http://arxiv.org/abs/2501.19362}.

\bibitem[BL21]{BinzLampart.2021}
T.~Binz and J.~Lampart.
\newblock An abstract framework for interior-boundary conditions.
\newblock Preprint, 2021,
  \burlalt{arXiv:2103.17124}{http://arxiv.org/abs/2103.17124}.

\bibitem[Can71]{Cannon.1971}
J.~Cannon.
\newblock Quantum field theoretic properties of a model of {N}elson: Domain and
  eigenvector stability for perturbed linear operators.
\newblock {\em J. Funct. Anal.}, 8(1):101--152, 1971.
\newblock \doi{10.1016/0022-1236(71)90023-1}.

\bibitem[Der03]{Derezinski.2003}
J.~Derezi{\'n}ski.
\newblock Van {H}ove {H}amiltonians -- Exactly Solvable Models of the Infrared
  and Ultraviolet Problem.
\newblock {\em Ann. Henri Poincar\'e}, 4(4):713--738, 2003.
\newblock \doi{10.1007/s00023-003-0145-5}.

\bibitem[DH22]{DamHinrichs.2021}
T.~N. Dam and B.~Hinrichs.
\newblock Absence of ground states in the renormalized massless
  translation-invariant {N}elson model.
\newblock {\em Rev. Math. Phys.}, 34(10):2250033, 2022,
  \burlalt{arXiv:1909.07661}{http://arxiv.org/abs/1909.07661}.
\newblock \doi{10.1142/S0129055X22500337}.

\bibitem[DM20a]{DamMoller.2018b}
T.~N. Dam and J.~S. M\o{}ller.
\newblock Asymptotics in Spin-Boson type models.
\newblock {\em Commun. Math. Phys.}, 374(3):1389--1415, 2020,
  \burlalt{arXiv:1808.00085}{http://arxiv.org/abs/1808.00085}.
\newblock \doi{10.1007/s00220-020-03685-5}.

\bibitem[DM20b]{DamMoller.2018a}
T.~N. Dam and J.~S. M\o{}ller.
\newblock Spin-Boson type models analysed using symmetries.
\newblock {\em Kyoto J. Math.}, 60(4):1261--1332, 2020,
  \burlalt{arXiv:1803.05812}{http://arxiv.org/abs/1803.05812}.
\newblock \doi{10.1215/21562261-2019-0062}.

\bibitem[DP14]{DeckertPizzo.2014}
D.-A. Deckert and A.~Pizzo.
\newblock Ultraviolet Properties of the Spinless, One-Particle {Y}ukawa Model.
\newblock {\em Commun. Math. Phys.}, 327(3):887--920, 2014,
  \burlalt{arXiv:1208.2646}{http://arxiv.org/abs/1208.2646}.
\newblock \doi{10.1007/s00220-013-1877-9}.

\bibitem[Eck70]{Eckmann.1970}
J.-P. Eckmann.
\newblock A model with persistent vacuum.
\newblock {\em Commun. Math. Phys.}, 18:247--264, 1970.
\newblock \doi{10.1007/BF01649435}.

\bibitem[GHL14]{GubinelliHiroshimaLorinczi.2014}
M.~Gubinelli, F.~Hiroshima, and J.~L\H{o}rinczi.
\newblock Ultraviolet renormalization of the {N}elson {H}amiltonian through
  functional integration.
\newblock {\em J. Funct. Anal.}, 267(9):3125--3153, 2014,
  \burlalt{arXiv:1304.6662}{http://arxiv.org/abs/1304.6662}.
\newblock \doi{10.1016/j.jfa.2014.08.002}.

\bibitem[GW16]{GriesemerWuensch.2016}
M.~Griesemer and A.~W{\"u}nsch.
\newblock On the domain of the {F}r{\"o}hlich {H}amiltonian.
\newblock {\em J. Math. Phys.}, 57(2):021902, 2016,
  \burlalt{arXiv:1508.02533}{http://arxiv.org/abs/1508.02533}.
\newblock \doi{10.1063/1.4941561}.

\bibitem[GW18]{GriesemerWuensch.2018}
M.~Griesemer and A.~W{\"u}nsch.
\newblock On the domain of the {N}elson {H}amiltonian.
\newblock {\em J. Math. Phys.}, 59(4):042111, 2018,
  \burlalt{arXiv:1711.10916}{http://arxiv.org/abs/1711.10916}.
\newblock \doi{10.1063/1.5018579}.

\bibitem[HH11]{HaslerHerbst.2010}
D.~Hasler and I.~Herbst.
\newblock Ground States in the Spin Boson Model.
\newblock {\em Ann. Henri Poincar\'e}, 12(4):621--677, 2011,
  \burlalt{arXiv:1003.5923}{http://arxiv.org/abs/1003.5923}.
\newblock \doi{10.1007/s00023-011-0091-6}.

\bibitem[HHL14]{HirokawaHiroshimaLorinczi.2014}
M.~Hirokawa, F.~Hiroshima, and J.~L\H{o}rinczi.
\newblock Spin-boson model through a {P}oisson-driven stochastic process.
\newblock {\em Math. Z.}, 277(3):1165--1198, 2014,
  \burlalt{arXiv:1209.5521}{http://arxiv.org/abs/1209.5521}.
\newblock \doi{10.1007/s00209-014-1299-1}.

\bibitem[HHS21]{HaslerHinrichsSiebert.2021a}
D.~Hasler, B.~Hinrichs, and O.~Siebert.
\newblock On Existence of Ground States in the Spin Boson Model.
\newblock {\em Commun. Math. Phys.}, 388(1):419--433, 2021,
  \burlalt{arXiv:2102.13373}{http://arxiv.org/abs/2102.13373}.
\newblock \doi{10.1007/s00220-021-04185-w}.

\bibitem[HHS22]{HaslerHinrichsSiebert.2021c}
D.~Hasler, B.~Hinrichs, and O.~Siebert.
\newblock {FKN} Formula and Ground State Energy for the Spin Boson Model with
  External Magnetic Field.
\newblock {\em Ann. Henri Poincar\'e}, 23(8):2819--2853, 2022,
  \burlalt{arXiv:2106.08659}{http://arxiv.org/abs/2106.08659}.
\newblock \doi{10.1007/s00023-022-01160-6}.

\bibitem[HL24]{HinrichsLampart.2023}
B.~Hinrichs and J.~Lampart.
\newblock A Lower Bound on the Critical Momentum of an Impurity in a
  Bose--Einstein Condensate.
\newblock {\em C.\,R. Math.}, 362:1399--1412, 2024,
  \burlalt{arXiv:2311.05361}{http://arxiv.org/abs/2311.05361}.
\newblock \doi{10.5802/crmath.652}.

\bibitem[HM22]{HiroshimaMatte.2019}
F.~Hiroshima and O.~Matte.
\newblock Ground states and their associated path measures in the renormalized
  {N}elson model.
\newblock {\em Rev. Math. Phys.}, 34(2):2250002, 2022,
  \burlalt{arXiv:1903.12024}{http://arxiv.org/abs/1903.12024}.
\newblock \doi{10.1142/S0129055X22500027}.

\bibitem[HM23]{HinrichsMatte.2023}
B.~Hinrichs and O.~Matte.
\newblock {F}eynman--{K}ac formula for fiber {H}amiltonians in the relativistic
  {N}elson model in two spatial dimensions.
\newblock Preprint, 2023,
  \burlalt{arXiv:2309.09005}{http://arxiv.org/abs/2309.09005}.

\bibitem[HM24a]{HinrichsMatte.2022}
B.~Hinrichs and O.~Matte.
\newblock {F}eynman--{K}ac Formula and Asymptotic Behavior of the Minimal
  Energy for the Relativistic {N}elson Model in Two Spatial Dimensions.
\newblock {\em Ann. Henri Poincar\'e}, 25(6):2877--2940, 2024,
  \burlalt{arXiv:2211.14046}{http://arxiv.org/abs/2211.14046}.
\newblock \doi{10.1007/s00023-023-01369-z}.

\bibitem[HM24b]{HinrichsMatte.2024}
B.~Hinrichs and O.~Matte.
\newblock {F}eynman--{K}ac formulas for semigroups generated by multi-polaron
  {H}amiltonians in magnetic fields and on general domains.
\newblock Preprint, 2024,
  \burlalt{arXiv:2403.12147}{http://arxiv.org/abs/2403.12147}.

\bibitem[HS95]{HuebnerSpohn.1995b}
M.~H\"{u}bner and H.~Spohn.
\newblock Radiative decay: nonperturbative approaches.
\newblock {\em Rev. Math. Phys.}, 7(3):363--387, 1995.
\newblock \doi{10.1142/S0129055X95000165}.

\bibitem[KP66]{KreinPetunin.1966}
S.~G. Kre{\u\i}n and Y.~I. Petunin.
\newblock Scales of {B}anach spaces.
\newblock {\em Uspehi Mat. Nauk}, 21(2):89--159, 1966.
\newblock \doi{10.1070/RM1966v021n02ABEH004151}.

\bibitem[Lam20]{Lampart.2019}
J.~Lampart.
\newblock The renormalized {B}ogoliubov-{F}r\"{o}hlich {H}amiltonian.
\newblock {\em J. Math. Phys.}, 61(10):101902, 2020,
  \burlalt{arXiv:1909.02430}{http://arxiv.org/abs/1909.02430}.
\newblock \doi{10.1063/5.0014217}.

\bibitem[Lam23]{Lampart.2023}
J.~Lampart.
\newblock Hamiltonians for Polaron Models with Subcritical Ultraviolet
  Singularities.
\newblock {\em Ann. Henri Poincar{\'e}}, 24(8):2687--2728, 2023,
  \burlalt{arXiv:2203.07253}{http://arxiv.org/abs/2203.07253}.
\newblock \doi{10.1007/s00023-023-01285-2}.

\bibitem[LCD{\etalchar{+}}87]{Leggettetal.1987}
A.~J. Leggett, S.~Chakravarty, A.~T. Dorsey, M.~P.~A. Fisher, A.~Garg, and
  W.~Zwerger.
\newblock Dynamics of the dissipative two-state system.
\newblock {\em Rev. Mod. Phys.}, 59(1):1--85, 1987.

\bibitem[LL05]{LiebLoss.2005}
E.~H. Lieb and M.~Loss.
\newblock Self-Energy of Electrons in Non-Perturbative QED.
\newblock In W.~Thirring, editor, {\em The Stability of Matter: From Atoms to
  Stars: Selecta of Elliott H. Lieb}, pages 607--623, Berlin, 2005. Springer,
  \burlalt{arXiv:math-ph/9908020}{http://arxiv.org/abs/math-ph/9908020}.
\newblock \doi{10.1007/3-540-27056-6\_40}.

\bibitem[LL23]{LillLonigro.2024}
S.~Lill and D.~Lonigro.
\newblock Self-adjointness and domain of generalized spin-boson models with
  mild ultraviolet divergences.
\newblock Preprint, 2023,
  \burlalt{arXiv:2307.14727}{http://arxiv.org/abs/2307.14727}.

\bibitem[Lon22]{Lonigro.2022}
D.~Lonigro.
\newblock Generalized spin-boson models with non-normalizable form factors.
\newblock {\em J. Math. Phys.}, 63(7):072105, 2022,
  \burlalt{arXiv:2111.06121}{http://arxiv.org/abs/2111.06121}.
\newblock \doi{10.1063/5.0085576}.

\bibitem[Lon23]{Lonigro.2023}
D.~Lonigro.
\newblock Self-Adjointness of a Class of Multi-Spin-Boson Models with
  Ultraviolet Divergences.
\newblock {\em Math. Phys. Anal. Geom.}, 26(2):15, 2023,
  \burlalt{arXiv:2301.10694}{http://arxiv.org/abs/2301.10694}.
\newblock \doi{10.1007/s11040-023-09457-6}.

\bibitem[LS19]{LampartSchmidt.2019}
J.~Lampart and J.~Schmidt.
\newblock On {N}elson-Type {H}amiltonians and Abstract Boundary Conditions.
\newblock {\em Commun. Math. Phys.}, 367(2):629--663, 2019,
  \burlalt{arXiv:1803.00872}{http://arxiv.org/abs/1803.00872}.
\newblock \doi{10.1007/s00220-019-03294-x}.

\bibitem[LSTT18]{LampartSchmidtTeufelTumulka.2018}
J.~Lampart, J.~Schmidt, S.~Teufel, and R.~Tumulka.
\newblock Particle Creation at a Point Source by Means of Interior-Boundary
  Conditions.
\newblock {\em Math. Phys. Anal. Geom.}, 21(2):12, 2018,
  \burlalt{arXiv:1703.04476}{http://arxiv.org/abs/1703.04476}.
\newblock \doi{10.1007/s11040-018-9270-8}.

\bibitem[MM18]{MatteMoller.2018}
O.~Matte and J.~S. M\o{}ller.
\newblock {F}eynman-{K}ac Formulas for the Ultra-Violet Renormalized {N}elson
  Model.
\newblock {\em Ast\'{e}risque}, 404, 2018,
  \burlalt{arXiv:1701.02600}{http://arxiv.org/abs/1701.02600}.
\newblock \doi{10.24033/ast.1054}.

\bibitem[Nel64]{Nelson.1964}
E.~Nelson.
\newblock Interaction of Nonrelativistic Particles with a Quantized Scalar
  Field.
\newblock {\em J. Math. Phys.}, 5(9):1190--1197, 1964.
\newblock \doi{10.1063/1.1704225}.

\bibitem[Par92]{Parthasarathy.1992}
K.~R. Parthasarathy.
\newblock {\em An Introduction to Quantum Stochastic Calculus}, volume~85 of
  {\em Monographs in Mathematics}.
\newblock Birkh\"{a}user, Basel, 1992.
\newblock \doi{10.1007/978-3-0348-0566-7}.

\bibitem[Pos20]{Posilicano.2020}
A.~Posilicano.
\newblock On the Self-Adjointness of {$H+A^*+A$}.
\newblock {\em Math. Phys. Anal. Geom.}, 23(4):37, 2020,
  \burlalt{arXiv:2003.05412}{http://arxiv.org/abs/2003.05412}.
\newblock \doi{10.1007/s11040-020-09359-x}.

\bibitem[Pos24]{Posilicano.2024}
A.~Posilicano.
\newblock On the Resolvent of {$H+A^*+A$}.
\newblock {\em Math. Phys. Anal. Geom.}, 27(3):11, 2024,
  \burlalt{arXiv:2307.13830}{http://arxiv.org/abs/2307.13830}.
\newblock \doi{10.1007/s11040-024-09481-0}.

\bibitem[RS72]{ReedSimon.1972}
M.~Reed and B.~Simon.
\newblock {\em Functional Analysis}, volume~1 of {\em Methods of Modern
  Mathematical Physics}.
\newblock Academic Press, San Diego, 1972.

\bibitem[Sch12]{Schmudgen.2012}
K.~Schm\"{u}dgen.
\newblock {\em Unbounded Self-adjoint Operators on {H}ilbert Space}, volume 265
  of {\em Graduate Texts in Mathematics}.
\newblock Springer, Dordrecht, 2012.
\newblock \doi{10.1007/978-94-007-4753-1}.

\bibitem[Sch19]{Schmidt.2019}
J.~Schmidt.
\newblock On a direct description of pseudorelativistic {N}elson
  {H}amiltonians.
\newblock {\em J. Math. Phys.}, 60(10):102303, 2019,
  \burlalt{arXiv:1810.03313}{http://arxiv.org/abs/1810.03313}.
\newblock \doi{10.1063/1.5109640}.

\bibitem[Sch21]{Schmidt.2021}
J.~Schmidt.
\newblock The Massless {N}elson {H}amiltonian and its Domain.
\newblock In A.~Michelangeli, editor, {\em Mathematical Challenges of
  Zero-Range Physics}, volume~42 of {\em Springer INdAM Series}. Springer,
  2021, \burlalt{arXiv:1901.05751}{http://arxiv.org/abs/1901.05751}.
\newblock \doi{10.1007/978-3-030-60453-0}.

\bibitem[Slo74]{Sloan.1974}
A.~D. Sloan.
\newblock The polaron without cutoffs in two space dimensions.
\newblock {\em J. Math. Phys.}, 15:190--201, 1974.
\newblock \doi{10.1063/1.1666620}.

\bibitem[Spo89]{Spohn.1989}
H.~Spohn.
\newblock Ground State(s) of the Spin-Boson {H}amiltonian.
\newblock {\em Commun. Math. Phys.}, 123(2):277--304, 1989.
\newblock \doi{10.1007/BF01238859}.

\bibitem[VH52]{vanHove.1952}
L.~Van~Hove.
\newblock Les difficult\'es de divergences pour un modelle particulier de champ
  quantifi\'e.
\newblock {\em Physica}, 18:145--159, 1952.
\newblock \doi{10.1016/S0031-8914(52)80017-5}.

\bibitem[Yaf92]{Yafaev.1992a}
D.~R. Yafaev.
\newblock On a zero-range interaction of a quantum particle with the vacuum.
\newblock {\em J. Phys. A: Math. Gen.}, 25(4):963--978, 1992.
\newblock \doi{10.1088/0305-4470/25/4/031}.

\end{thebibliography}

\end{document}